\documentclass[10pt]{IEEEtran}

\pdfoutput=1  

\IEEEoverridecommandlockouts                              
		                                                  
\usepackage{flushend}

\usepackage{cite,url}

\usepackage[usenames,dvipsnames]{color}
\usepackage{graphicx,wrapfig,subfigure}

\usepackage[cmex10]{amsmath}
\usepackage{amsfonts,amssymb,mathrsfs}   
\usepackage[colorlinks,bookmarksopen,bookmarksnumbered,citecolor=red,urlcolor=red]{hyperref}
\usepackage{soul}

\newlength{\noteWidth}
\long\def\notes#1{\ifinner
           {\footnotesize #1}
           \else
           \marginpar{\parbox[t]{\noteWidth}{\raggedright\footnotesize #1}}
       \fi\typeout{#1}}

\def\notes#1{\typeout{read notes: #1}}  

\def\spm#1{\notes{SPM:  #1}}

\def\ptcl{\varrho^{\text{etr}}}
\def\thetamb{\theta^0}

\def\Thetamin{\Theta^{\text{min}}}
\def\Thetamax{\Theta^{\text{max}}}

\def\sq{\hbox{\rlap{$\sqcap$}$\sqcup$}}
\def\qed{\ifmmode\sq\else{\unskip\nobreak\hfil
\penalty50\hskip1em\null\nobreak\hfil\sq
\parfillskip=0pt\finalhyphendemerits=0\endgraf}\fi\medskip}

\long\def\defbox#1{\framebox[.9\hsize][c]{\parbox{.85\hsize}{%
\parindent=0pt
\baselineskip=12pt plus .1pt      
\parskip=6pt plus 1.5pt minus 1pt 
 #1}}}

\long\def\beginbox#1\endbox{\subsection*{}%
\hbox{\hspace{.05\hsize}\defbox{\medskip#1\bigskip}}%
\subsection*{}}

\def\endbox{}

\def\transpose{{\hbox{\it\tiny T}}}

\newsavebox{\junk}
\savebox{\junk}[1.6mm]{\hbox{$|\!|\!|$}}

\def\state{{\sf X}}

\newcommand{\field}[1]{\mathbb{#1}}

\def\Re{\field{R}}

\def\One{\mbox{\rm{\large{1}}}}

\def\ind{\field{I}}

\def\cP{{\check{P}}}

\def\bfmath#1{{\mathchoice{\mbox{\boldmath$#1$}}%
{\mbox{\boldmath$#1$}}%
{\mbox{\boldmath$\scriptstyle#1$}}%
{\mbox{\boldmath$\scriptscriptstyle#1$}}}}

\def\bfmr{\bfmath{r}}

\def\bfmX{\bfmath{X}}

\def\bfmY{\bfmath{Y}}

\def\bfmhhaY{\bfmath{\hhaY}} 
\def\bfmhhaY{\hbox to 0pt{$\widehat{\bfmY}$\hss}\widehat{\phantom{\raise 1.25pt\hbox{$\bfmY$}}}}

\def\bfmV{\bfmath{V}}  
\def\bfmW{\bfmath{W}}  
  
\def\bfmZ{\bfmath{Z}}

\def\bfPhi{\bfmath{\Phi}}

\def\bfPsi{\bfmath{\Psi}}

\def\bfzeta{\bfmath{\zeta}}

\def\til={{\widetilde =}}
\def\tilPhi{{\widetilde \Phi}}

\def\clI{{\cal I}}

\def\clO{{\cal O}}

\def\clV{{\cal V}}

\def\clY{{\cal Y}}

 \def\FRAC#1#2#3{\genfrac{}{}{}{#1}{#2}{#3}}

\def\ddtp{{\mathchoice{\FRAC{1}{d^{\hbox to 2pt{\rm\tiny +\hss}}}{dt}}%
{\FRAC{1}{d^{\hbox to 2pt{\rm\tiny +\hss}}}{dt}}%
{\FRAC{3}{d^{\hbox to 2pt{\rm\tiny +\hss}}}{dt}}%
{\FRAC{3}{d^{\hbox to 2pt{\rm\tiny +\hss}}}{dt}}}}

\def\half{{\mathchoice{\FRAC{1}{1}{2}}%
{\FRAC{1}{1}{2}}%
{\FRAC{3}{1}{2}}%
{\FRAC{3}{1}{2}}}}

\def\eqdef{\mathbin{:=}}

\def\Prob{{\sf P}}

\def\Expect{{\sf E}}

\def\average#1,#2,{{1\over #2} \sum_{#1}^{#2}}

\def\eye(#1){{\bf(#1)}\quad}

\newtheorem{theorem}{Theorem}[section]

\newtheorem{proposition}[theorem]{Proposition}

\def\Proposition#1{Prop.~\ref{#1}}

\def\Section#1{Section~\ref{#1}}

\def\eq#1/{(\ref{e:#1})}

\newcommand{\beqn}[1]{\notes{#1}%
\begin{eqnarray} \elabel{#1}}

\newcommand{\eeqn}{\end{eqnarray} }

\newcommand{\beq}[1]{\notes{#1}%
\begin{equation}\elabel{#1}}

\newcommand{\eeq}{\end{equation}}

\def\bdes{\begin{description}}
\def\edes{\end{description}}

\def\bary{{\overline {y}}}

\def\barSigma{\overline{\Sigma}}

\newcounter{rmnum}
\newenvironment{romannum}{\begin{list}{{\upshape (\roman{rmnum})}}{\usecounter{rmnum}
\setlength{\leftmargin}{8pt}
\setlength{\rightmargin}{6pt}
\setlength{\itemsep}{2pt}
\setlength{\itemindent}{-2pt}
}}{\end{list}}

\newcounter{anum}

{\end{list}}

\def\ass(#1:#2){(#1\ref{#1:#2})}

\def\ritem#1{
\item[{\sf \ass(\current_model:#1)}]
}

\newenvironment{recall-ass}[1]{%
\begin{description}
\def\current_model{#1}}{
\end{description}
}

\newcommand{\bd}{\begin{description}}
\newcommand{\ed}{\end{description}}
\newcommand{\bt}{\begin{theorem}}
\newcommand{\et}{\end{theorem}}
\newcommand{\ba}{\begin{array}{rcl}}
\newcommand{\ea}{\end{array}}

\def\Proposition#1{Prop.~\ref{#1}}
\def\Prop#1{Prop.~\ref{#1}}

\def\haY{\widehat Y}
\def\haPhi{\widehat \Phi}

\def\util{\mathchoice{\mbox{\small$\cal U$}}%
{\mbox{\small$\cal U$}}%
{\mbox{$\scriptstyle\cal U$}}%
{\mbox{$\scriptscriptstyle\cal U$}}}

\def\bfgamma{\bfmath{\gamma}}

\def\diag{\,\text{\rm diag}\,}   

\def\Spx{\textsf{S}}

\graphicspath{{figures/}}

\def\Ebox#1#2{%
\begin{center}
\includegraphics[width= #1\hsize]{#2} 
\end{center}}

\def\tilY{\widetilde{Y}}

\def\Fig#1{Fig.~\ref{#1}}

\def\ind{\field{I}}

\def\Re{\field{R}}

\def\piload{\Gamma}

\def\Health{{\cal L}}
\def\health{\ell}

\title{State Estimation for the Individual and the Population in Mean Field Control 
\\
with Application to Demand Dispatch  
}

\author{Yue Chen, 
Ana Bu\v{s}i\'c, and Sean Meyn
\thanks{Research   supported by  NSF grants CPS-0931416 and CPS-1259040,
and the French National Research Agency grant ANR-12-MONU-0019}
\thanks{Y.C. and S.M. are with the Department of Electrical and Computer
Engg.\ at the University of Florida, Gainesville. A.B.\ is with Inria and the Computer Science Dept. of \'Ecole Normale Sup\'erieure, Paris, France.}%
}

\begin{document}

\maketitle

\begin{abstract}

This paper concerns state estimation problems in a mean field control setting. In a finite population model, the goal is to estimate the joint distribution of the population state and the state of a typical individual. The observation equations are a noisy measurement of the population. 

The general results are applied to demand dispatch for regulation of the power grid, based on randomized local control algorithms. In prior work by the authors it is shown that local control can be designed so that the aggregate of loads behaves as a controllable resource, with accuracy matching or exceeding traditional sources of frequency regulation. The operational cost is nearly zero in many cases. 

The information exchange between grid and load is minimal, but it is assumed in the overall control architecture that the aggregate power consumption of loads is available to the grid operator. It is shown that the Kalman filter can be constructed to reduce these communication requirements, and to provide the grid operator with accurate estimates of the mean and variance of quality of service (QoS) for an individual load.


\end{abstract}
 

\section{Introduction}

\notes{Given the new title, we need to say that the methods go beyond power or pools!}
\notes{Ana: Also, the paper may now go to reviewers that do not know anything on power. I added two sentences to explain what is it that we control.}

Mean field models are a valuable tool for design and performance approximation for certain classes of interacting systems \cite{huacaimal07,matkoccal13,chebusmey14}.  The infinite-population mean-field equations provide tremendous insight, but ultimately we must translate this insight to address a finite-population reality.   In this paper we propose algorithms based on the Kalman filter to obtain estimates of first and second order statistics of the population and a typical individual, based on noisy observations of the population.    

While the potential applications are far broader than power systems,   for ease of exposition it is convenient to restrict attention to one application.

Renewable energy sources such as wind and solar power have a high degree of unpredictability and time variation, which complicates balancing supply and demand. One possible way to address this challenge is to harness the inherent flexibility in demand of many types of loads.

\textit{Demand Response} is traditionally meant as a reduction in load in response to some grid-level event.   It is in use today for peak-shaving (smoothing demand),  and for contingency reserves (load-shedding following generation loss).   

It is argued in \cite{calhis11,barbusmey14,meybarbusyueehr15} that the value of demand-side flexibility is far greater than this.   
Loads can supply a range of grid services, such as the balancing reserves required at BPA,  or the Reg-D/A regulation reserves used at PJM \cite{barbusmey14}.  These grid services can be obtained without impacting quality of service (QoS)  for consumers \cite{chebusmey14,haolinkowbarmey14}.   
This is only possible through design. 
The term \textit{Demand Dispatch}, introduced in \cite{brolureispiwei10}, is used 
to emphasize the difference between the goals of our own work and traditional demand response. 

%
%
%
%
%
%

The application in this  paper concerns a large collection of loads whose power consumption is not continuously variable.  Examples include thermostatically controlled loads (TCLs), as considered in  \cite{matkoccal13,kizmal14a},  and irrigation or pool-pumps \cite{meybarbusyueehr15,chebusmey14}.  In these   papers 
and
\cite{chrtomlebpao14,lumThesis15}
it is argued that randomization at the load is valuable to avoid synchronization, and to simplify control at the grid level.   

In much of this prior work, a mean-field model is obtained for control design at the grid level -- this is a deterministic model of the aggregate of loads,  obtained as a law of large numbers limit as the population of loads tends to infinity.   The control solution adopted in \cite{matkoccal13} is based on state-feedback for a linear state space model with partial observations.  Although the mean-field model is bi-linear,   it is represented as a linear model by treating the product of inputs and states as a new input; 
this is why state estimation is needed for implementation of the algorithm.  In  \cite{meybarbusyueehr15} a randomized policy is designed for each load so that the mean-field model is an input-output system that is easily controlled without the use of state estimates.   

\notes{Johanna told me she tried out "my estimator"   - what has JM done?  \\
See this paper Raginsky suggested,  cite{berwu98}.
 ...
 I doubt it is useful
 \\
 Also look at
Kalman filtering in triplet Markov chains  (IEEE Signal processing)
Generalized dynamic linear models for financial time series
   ... and ...
   Diffusion Approximation for Bayesian Markov Chains
} 

For simplicity, in this paper   attention is restricted to the setting of   \cite{meybarbusyueehr15}, in which each load evolves as a controlled Markov chain.   The transition probability is determined by its own state, and  a  scalar signal $\bfzeta$ broadcast from a balancing authority (BA).   The extension to vector inputs, as in \cite{matkoccal13}, requires only changes in notation.   
    
The common dynamics are defined by a controlled transition matrix $\{P_\zeta : \zeta\in\Re\}$.  For the $i$th load, there is a state process $\bfmX^i$ whose transition probability is defined by,
\begin{equation}
P_{\zeta}(x^-,x^+) 
=
\Prob\{X^i_{t+1} = x^+ \mid X^i_t = x^- ,\,  \zeta_t=\zeta \}  
\label{e:Pzeta}
\end{equation}
where $x^-$ and $x^+$ are possible  state-values.  In the case of a water heater, the state $x\in\state$ might represent temperature of the water, and whether the unit is operating or not.  

If there are $N$ loads operating independently, conditional on the common signal $\bfzeta$, then
the empirical distribution (i.e., the histogram of state values) is defined as the average,
\[	 
\mu_t^N( x) 
=\frac{1}{N}\sum_{i=1}^N  \ind\{X^i_t =x \} 
\, , \quad x\in\state
\]
Viewed as a row vector, the following recursion is central to the analysis in \cite{meybarbusyueehr15,chebusmey14}:
\begin{equation}
	 \mu_{t+1}^N = \mu_t^N  P_{\zeta_t}+W_{t+1}^\transpose\, .
	\label{e:empir_dist}
\end{equation}
 An observation model  is also linear in the state,
\begin{equation}
	Y_t = \sum_x\mu_t^N (x)\util(x)  + V_t
	\label{e:obs}
\end{equation}
where $\util\colon\state\to\Re$.  
In applications to demand dispatch,   $\util(x)$ represents power consumption of a load when its state is $x$, so that $\sum_x\mu_t^N (x)\util(x)$ is the average power consumption at time $t$.

 It is established in \cite{meybarbusyueehr15,chebusmey14} that  
$\bfmW\eqdef \{W_t : t\ge 1\}$ is a $d$-dimensional martingale-difference sequence (and hence uncorrelated).  The i.i.d.\ sampling model assumed in the present paper implies that  $\bfmV=\{V_t : t\ge 1\}$ is a martingale-difference sequence  that is also uncorrelated with $\bfmW$.

 \Prop{t:PhiLin} contains full details of this system description.  
 
The Kalman filter is developed in two settings, each with the same observation process:
 The first is constructed to obtain    estimates of $\mu_t^N$.  The second filter obtains estimates of the joint statistics of a larger state that includes  both $\mu_t^N$ and  the state of a typical individual.  
 The main conclusions are summarized here:
\begin{romannum}

\item  A measurement architecture is proposed in which each load broadcasts its state only occasionally -- say, once per day.   The observation equations in the aggregate model then include white noise, whose conditional variance is computed.  The state equations for the population/individual  dynamics evolve as a linear stochastic system with white noise disturbance, whose conditional covariance matrix is also computable. 


\spm{At Illinois we always used Grammian -- I don't know why Wikipedia uses just one 'm'}
\item  In the examples considered,  the observability Grammian is not full rank, and an approximate time-invariant model is also  unobservable.  Moreover, \Prop{t:symNot} demonstrates that a symmetry property (that holds in the example of \cite{meybarbusyueehr15}),
implies that the model cannot be observable.

However, in numerical experiments it is found that the Kalman filter remains valuable for reducing the impact of measurement noise,  and for estimating the distribution of quality of service.    
   In particular,   estimates of certain first and second order statistics of an individual load are remarkably accurate, even though the measurements are a noisy sequence of samples from the population.
   
\item      In the face of un-modeled dynamics such as load heterogeneity, or additional ``opt-out'' control used to enforce QoS bounds,  the Kalman filter combined with PI control continues to perform nearly perfectly, even with 0.1\%\ sampling of loads.
\end{romannum}

%
%
%
%

There are in fact two general formulations of the Kalman filter.  In the first, most typical setting,  the sequence of Kalman gains is deterministic;  obtained through a Riccati equation (a recursive equation driven by the covariance matrices for the state and observation noise).   This is known to be $L_2$-optimal over all estimators that are \textit{linear} functions of the observations.  

A second formulation of the Kalman filter uses \textit{conditional} covariance matrices to define the Kalman gain  --- see (\ref{e:NoiseCondCov},\ref{e:NoiseObsCondCov})
 and surrounding discussion.   If the state/observation noise is conditionally Gaussian,  then this  Kalman filter coincides with the nonlinear filter, which is $L_2$-optimal over all causal estimators~\cite{cai88}.   Because the Riccati equation is a nonlinear function of the covariance matrices,   this version of the Kalman filter may be a nonlinear function of the observations.
 
The second  is attractive because it is easy to compute formulae for the conditional covariance matrices,  while the unconditional covariance matrices only admit approximations.   Moreover, when considering the dynamics of the aggregate, a Gaussian approximation of the noise is justified by the Central Limit Theorem (CLT).

\paragraph*{Related research}
In addition to the references cited above, there are many papers on demand dispatch based on centralized control, or relying on real-time prices to solve the control problem of interest. Much of the latter is closer to demand response,  and has little intersection with the research summarized here.

An application of this formulation of the Kalman filter was considered previously in \cite{lipkrirub84} for a single Markov chain without control, with measurements subject to Gaussian error.  

\spm{Nov 2015:  I reread Montreal work, and it is just partially observed LQG}

There are several recent papers with similar goals in the literature on mean-field models.  Most closely related is \cite{caikiz13a,sencai14} which concerns  partially observed LQG mean field games, with several classes of players.   The state estimation problem is from the point of view of the individual -- each ``minor agent'' obtains noisy and partial observations,  and wishes to estimate the ``major state'' as well as the aggregate. The solution is obtained through the construction of a Kalman filter.  This prior work is also motivated by application to power systems.


\smallbreak

The remainder of the paper consists of four sections organized as follows.  The following section describes the stochastic model on which the estimation algorithms are based.  
Filtering equations are derived in \Section{s:Kalman} for a collection of 
controlled Markov models with partial observations.  
The algorithms have been tested in various different settings -- results for demand dispatch using residential pools and also  TCLs
are summarized in
\Section{s:num}.
Conclusions and directions for future research are contained in
\Section{s:con}.

Further results may be found in the dissertation \cite{YueChenThesis16}.

\section{Mean Field Model}
\label{s:mfm}

It is assumed throughout the paper that a family of Markov transition matrices $\{P_\zeta : \zeta\in\Re\}$ is given that is continuous in the parameter $\zeta$.   The finite state space is denoted $\state = \{x^1,\dots,x^d\}$, so that each $P_\zeta$ is a $d\times d$ matrix. 

The mean-field model is defined as the approximation of 	\eqref{e:empir_dist}
obtained as $N\to\infty$.  This is the deterministic recursion,  
\begin{equation}
	 \mu_{t+1}= \mu_t  P_{\zeta_t} ,
	\label{e:mfm}
\end{equation}
with $\mu_0$ given, and where $\bfzeta$ is obtained via causal feedback.   This paper is concerned with the stochastic system \eqref{e:empir_dist},  but the steps used to justify the
limit
 will lead to the second-order statistics required to describe the Kalman filter.

\subsection{Aggregate dynamics}

The individual dynamics are described by the controlled Markov model \eqref{e:Pzeta}.   We lift the state space from the $d$-element set
$\state = \{x^1,\cdots,x^d\}$,  to the $d$-dimensional simplex $\Spx$.  For the $i^{th}$ load at time $t$, the element $\piload_t^i  \in \Spx$ is the degenerate distribution whose mass is concentrated at $x$ if   $X^i_t= x$;
that is, $\piload_t^i =\delta_x$.  These distributions evolve according to a random linear system,
\begin{equation}
	\piload_{t+1}^i=\piload_t^iG_{t+1}^i 
\label{e:piG}
\end{equation}
in which $\piload_t^i$ is interpreted as a $d$-dimensional row vector,   $G_t^i$ is a $d\times d$ matrix with entries $0$ or $1$ only, and $\sum_l G_t^i(x^j,x^l)=1$ for all $j$.   

The following assumptions are imposed throughout.


\begin{romannum}
\item[A1:] \spm{new:}
The input is defined by causal output feedback:
For a continuous family of functions $\phi_t\colon \Re^{t+1}\to\Re$,  
 \[
\zeta_t = \phi_t(Y_0,\dots, Y_t)\,,\qquad t\ge 0.
\]

 \item[A2:] 
For some function $\Xi$ with domain $ \Re\times [0,1]$,
and range equal to the set of $d\times d$ matrices,   
\begin{equation}
 G_t^i=\Xi(\zeta_{t-1},\xi_t^i),
\label{e:A1}
\end{equation}
 where $\{ \xi_t^i: t\geq 1, \ i\ge 1\}$ are i.i.d.\ on $[0,1]^d$ with uniform distribution.

\item[A3:]  
The initial conditions $\{\piload_0^i : 1\le i\le N\}$ are i.i.d., with
\[
\Prob\{ \piload_0^i (x)=1\} = \mu_0(x),\quad x\in\state 
\] 
\item[A4:]  The measurements are obtained through sampling:   
There is a bounded sequence $\{ \gamma_t: t\geq 1 \}$ of $N$-dimensional vectors with non-negative entries, 
and independent of the   $\{ \xi_t^i\}$, such that
$\Expect [  \gamma_t(k) ] =N^{-1}$  for each $t$ and $k$, and
\[
Y_t =  \sum_{i=1}^N \gamma_t(i) \util(X^i_t). 
\]
Moreover,  the distribution of $\gamma$ is unchanged by permutations of its $N$ components.   
\end{romannum}
The first assumption is based on the control architecture assumed in all prior work.  
\spm{new:}
 Assumption~A2 is standard in this context \cite{lipkrirub84},  and can be assumed since any probability mass function (pmf) is a function of a uniformly distributed random variable.
 
Assumption A4 can be used to model random sampling with or without replacement. 
\spm{new:}
 Examples are given in \Section{s:num}.   This assumption leads to the measurement equation \eqref{e:obs}; additional additive noise can be included in the observation equation, provided this measurement noise is independent of the sampling process.

The filtration of observations is denoted,
\begin{equation}
\clY_t = \sigma\{ Y_r,\zeta_r : r\le t\},\quad t\ge 0.
\label{e:clY}
\end{equation}
Under A1--A3, $G_{t+1}^i $
 is conditionally independent of $\{\piload_0^i,\cdots,\piload_t^i\}$, given $\clY_t$, with
\begin{equation}
\label{e:EG=P} 
 	\Expect[G_{t+1}^i \mid \clY_t ]=P_{\zeta_t}.
\end{equation}

Two filtering problems are considered in this paper.  In the first, the state is equal to the empirical distributions expressed as column vectors: 
\begin{equation}
\Phi_t(k) = \mu^N_t(x^k) ,\qquad 1\le k\le d,   \ t\ge 0.
\label{e:Phi}
\end{equation} 
In the second,  the goal is to estimate the state for an individual load.  For the $i$th load, this is denoted
\begin{equation}
\Phi_t^i(k) =  \piload_t^i (x^k) ,\qquad 1\le k\le d,   \ t\ge 0.
\label{e:Phi-i}
\end{equation}
The two state processes are evidently related,
\begin{equation}
\Phi_t = \frac{1}{N} \sum_{i=1}^N \Phi_t^i
\label{e:Phi-sum}
\end{equation} 

\begin{proposition}
\label{t:PhiLin}
The two state processes each evolve as linear systems, 
\begin{eqnarray}
\Phi_{t+1} &=& A_t \Phi_t + W_{t+1}
\label{e:StateForEst}
\\
\Phi^i_{t+1} &=& A_t \Phi_t^i + W_{t+1}^i
\label{e:StateForEst-i}
\end{eqnarray}
where $A_t=P_{\zeta_t}^\transpose$,   
and for each $i$, $t$, 
\begin{eqnarray}
 W_{t+1}  &=& \frac{1}{N}\sum_{k=1}^N W_{t+1}^k
\label{e:W}
\\
W^i_{t+1} &=&
\bigl[
\piload_t^i(G_{t+1}^i -P_{\zeta_t}) \bigr]^\transpose
\label{e:Wi}
\end{eqnarray}
The observation equation \eqref{e:obs} can be written,
\begin{equation}
Y_t = C \Phi_t + V_t
\label{e:obsPhi}
\end{equation}
where $C_j = \util(x^j)$ for each $j$.    

Moreover,  each of the sequences $\bfmW^i$,  $\bfmW$, $\bfmV$ is a martingale difference sequence,  and the sequence $\bfmV$ is also uncorrelated with each of the  $\{\bfmW^i\}$ and also  $\bfmW$.
\end{proposition}

\begin{proof}
The martingale difference property for  each of the $\{\bfmW^i\}$  (and hence also $\bfmW$),
follows immediately from \eqref{e:EG=P};  see also \cite{chebusmey14}.
The  sequence $\bfmV$ can be expressed,
\[
V_t = - \sum_k \mu^N_t(x^k)  \util(x^k)  + \sum_{i=1}^N \gamma_t(i) \util(X^i_t)
\]
This has mean zero, from the exchangeability assumption in A4, and it is a martingale difference sequence because the sequence $\bfgamma$ is i.i.d.\    It is uncorrelated with $\{\bfmW^i\}$ and   $\bfmW$ under the assumption that $\{\gamma_t\}$ is independent of $\{\xi_t^i\}$.
 \end{proof}

  A  derivation of the conditional state covariances is
 given in \Section{s:KFstate}.
 
For a linear-Gaussian model, the Kalman filter equations are intended to approximate the conditional mean and covariance of the state. In the first model \eqref{e:StateForEst}
they are denoted,
\begin{equation}
\haPhi_t =  \Expect[\Phi_t  \mid\clY_t]\,, \quad
\Sigma_t = \Expect[\tilPhi_t \tilPhi_t^\transpose\mid\clY_t]
\label{e:meanW-e:SigPhi}
\end{equation}
with $\tilPhi_t =\Phi_t - \haPhi_t $.  For large $N$ it can be argued via the CLT that $(\bfmV,\bfmW)$ is approximately conditionally Gaussian given the observations.  We might expect the Kalman filter to approximate the optimal nonlinear filter in this case.   

In  \Section{s:EstInd}  the state space is extended to obtain estimates of QoS metrics for an individual load that may not be a function of the respective state $\Phi^i_t$.

\subsection{Example:  Intelligent pools}

Here we focus on a single example in which each load is a residential pool pump.  
\Section{s:num} contains
  extensions to other loads.  

In the original model of \cite{meybarbusyueehr15}, the state space is taken to be the finite set,
\begin{equation}
\state=\{  (m,k) :  m\in  \{  \oplus,\ominus\}  ,\  k \in \{1,\dots,\clI\} \}  
\label{e:poolstate}
\end{equation}
where $\clI>1$ is an integer.  If $X_t^i = (\oplus, k)$,  this means that the pool pump is on at time $t$, and has remained on for the past $k$ time units.   In this paper we take the same state space,  but with a new interpretation of each state.

\begin{figure}[h]
\Ebox{.75}{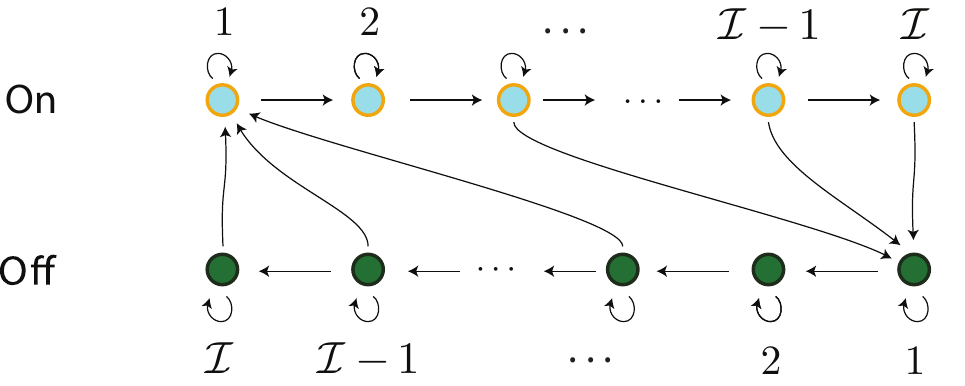} 
\vspace{-2.5ex}
\caption{State transition diagram for pool  model.}
\label{f:ppp}
\vspace{-1.25ex}
\end{figure}

The controlled transition matrix is of the form,
\begin{equation}
P_\zeta = (1-\delta)I + \delta \cP_\zeta
\label{e:poolP}
\end{equation}
in which $\cP_\zeta$ is the transition matrix used in \cite{meybarbusyueehr15},  and $\delta\in (0,1)$. At each time $t$, a weighted coin is flipped with probability of heads equal to $\delta$.  If the outcome is a tail,  then the   state does not change.  Otherwise,  a transition is made from the current state $x$ to a new state $x^+$ with probability $ \cP_{\zeta_t}(x,x^+)$.   

A state transition diagram is shown in \Fig{f:ppp}.  The state transition diagram for $\cP_\zeta$ is identical, except that the self-loops are absent.

\begin{figure}[h]
\Ebox{.95}{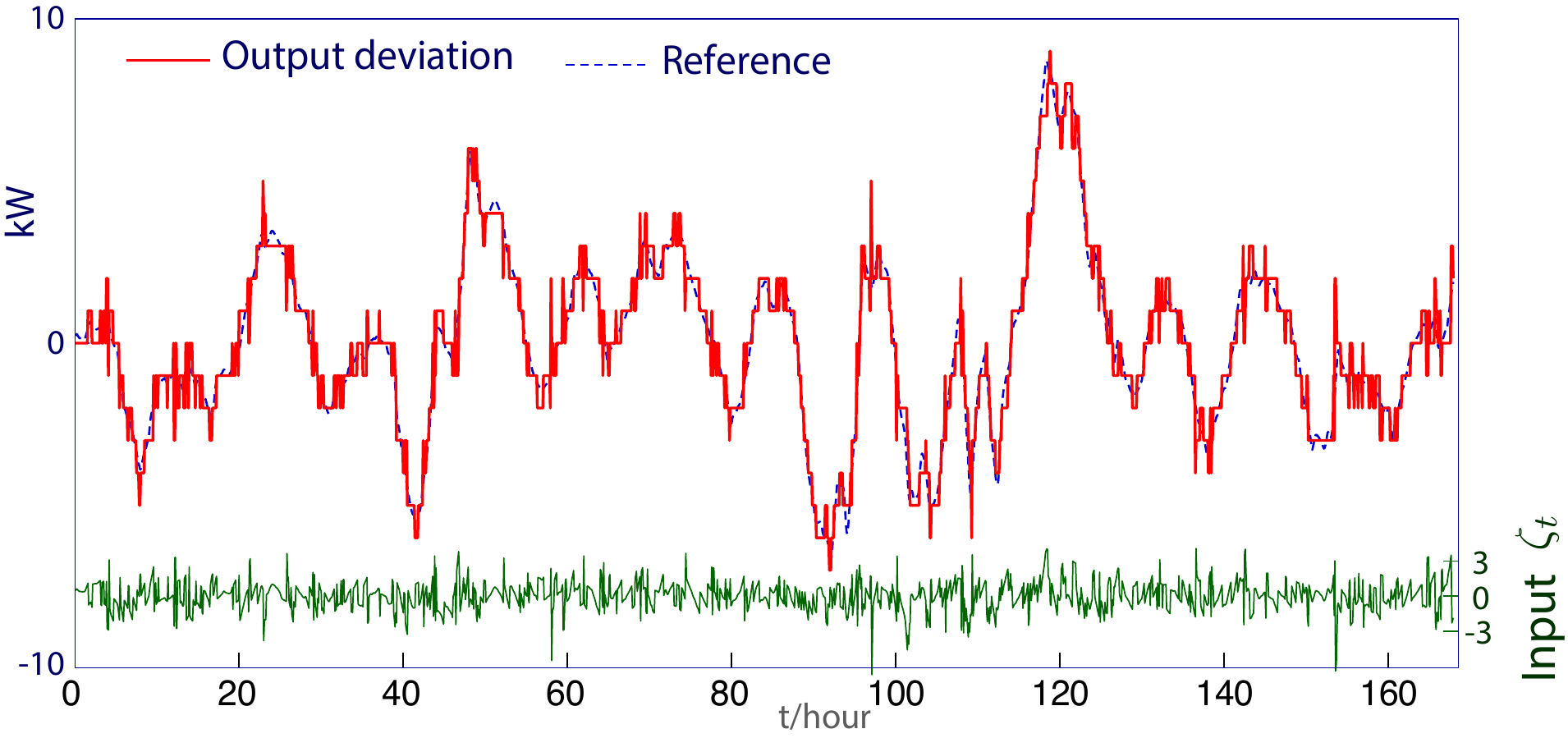} 
\vspace{-2.5ex}
\caption{The deviation in power consumption tracks well even with only  20 pools engaged.}
\label{f:20pools}
\vspace{-1.25ex}
\end{figure}

The motivation comes from conflicting needs of the grid and the load:   a single load turns on or off only a few times per day, yet the grid operator wishes to send a signal far more frequently -- In this example we assume every 5 minutes.   If the sampling increments for each load  were taken to be 5 minutes, then it would be necessary to take $\clI$ very large in the approach of \cite{meybarbusyueehr15}.

In this paper $\cP_\zeta$ is obtained using the optimal-control approach of  \cite{meybarbusyueehr15}; we take $\clI=48$, and hence $d=|\state|=96$.   It is assumed that $\delta=1/6$,   so that  the pool state changes every 30 minutes on average.   In \cite{meybarbusyueehr15} it is shown that the transition matrix has desirable properties for control:  the linearized dynamics are minimum phase,  with positive DC gain.   Hence, for example, 
a persistent positive value of $\zeta_t$ will lead to an increase in aggregate power consumption.

For sake of illustration,
\Fig{f:20pools} shows tracking performance of this scheme with only \textit{twenty pools}.     Each pool is assumed to consume 1~kW when operating, and each has a 12 hour cleaning cycle.   The grid operator uses PI compensation to determine $\bfzeta$ (see \Section{s:hetero}
 for details).  With such a small number of loads it is not surprising to see some evidence of quantization.  For 100 loads or more, and the reference scaled proportionately, the tracking is nearly perfect.

Two QoS metrics have been considered for this model.  First is `chattering' -- a large number of switches from on to off.  A large value means poor QoS, but this is already addressed through design of the controlled transition matrix \cite{meybarbusyueehr15}.   The design of $P_\zeta$ also helps to enforce  upper and lower bounds on the duration of cleaning each day.

A second metric is total cleaning over a time horizon of one week or more. This is the QoS metric considered in \cite{chebusmey14}.  
In this paper we consider a discounted version:
We assume that $P_0$ has a unique invariant pmf $\pi_0$.
With $\health\colon\state\to\Re$ a given function with zero steady-state mean,  $\sum_x \pi_0(x) \health(x) =0$,  we define for each $i$ and $t$,
\begin{equation}
\Health^i_t = \sum_{k=0}^t \beta^{t-k} \health(X^i_k)
\label{e:QoS15}
\end{equation} 
with $\beta\in (0,1]$ a constant.
The function $\health(x) = \ind(m=\oplus) - \ind(m=\ominus)$ was used in \cite{chebusmey14} in the case of a 12 hour cleaning cycle (recall the notation $x= (m,k) $).

Experimental results surveyed in \Section{s:num} demonstrate that it is possible to estimate functionals of the state process such as the QoS metric $\{\Health^i_t\}$, 
even if the observations are subject to   significant measurement noise.



%
%

\begin{figure}
\Ebox{.7}{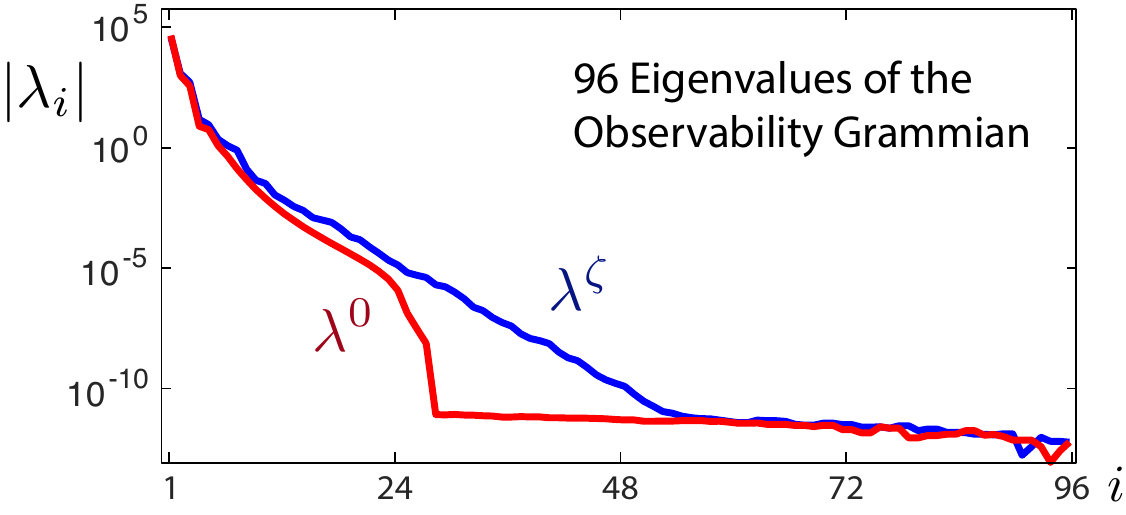} 
\vspace{-2.5ex}
\caption{Eigenvalues for the observability Grammian for the pool  model in two cases:   The magnitude of eigenvalues decays rapidly for a typical sample-path of $\bfzeta$,  and for the LTI model obtained with $\bfzeta\equiv 0$.}
\label{f:obsGramm}
\vspace{-1ex}
\end{figure}

In anticipation of the results to come we ask, \textit{what would linear systems theory predict with respect to state estimation performance?}  The observability Grammian associated with \eqref{e:StateForEst} was computed for typical sample paths $\{A_t = P_{\zeta_t}^\transpose : 1\le t \le 2016\}$, where the value $2016$ corresponds to one week,  and $\bfzeta$ was scaled to lie between $\pm \half$.  Its rank was found to be approximately 40,  while the maximal rank is 96 (the dimension of the state).   With $\zeta_t\equiv 0$ the system is time-invariant.  In this case the rank of the observability Grammian coincides with the rank of the observability matrix, which was found to be  23.

\spm{new:}
However, these values were obtained using the ``rank'' command in Matlab,  which relies on finite numerical precision.  A plot of the magnitude of the eigenvalues for the two observability Grammians shown in   \Fig{f:obsGramm} suggests that both matrices are full rank.   

\spm{new:}
Further analysis establishes that the LTI model obtained with $\bfzeta\equiv 0$ \textit{cannot be observable}, due to a particular symmetry found in this example.  A general result given in  \Prop{t:symNot} implies that $\lambda^0_i=0$ for $50\le i\le 96$.


\section{Kalman Filter Equations}
\label{s:Kalman}

The second order statistics for the disturbances appearing in the linear model (\ref{e:StateForEst}, \ref{e:StateForEst-i}, \ref{e:obsPhi}) are derived here.   These expressions are used
to construct a Kalman filter that generates approximations for the conditional mean and covariance
\eqref{e:meanW-e:SigPhi}.
Other statistics of interest are,
\[
\begin{aligned}
\haPhi_t^i &=   \Expect[\Phi^i_t | \clY_t] \,, \quad
 \Sigma^i_t    =  \Expect[\tilPhi^i_t(\tilPhi^i_t)^\transpose | \clY_t] 
\\[.1cm] 
\haPhi_{t+1\mid t}^i &=   \Expect[\Phi^i_{t+1} | \clY_t] \,, \quad
 \Sigma^i_{t+1\mid t}    =  \Expect[\tilPhi^i_{t+1}(\tilPhi^i_{t+1})^\transpose | \clY_t] 
\\[.1cm] 
\haPhi_{t+1\mid t} &=   \Expect[\Phi_{t+1} | \clY_t]\,, \quad
  \Sigma_{t+1\mid t}    =  \Expect[\tilPhi_{t+1}(\tilPhi_{t+1})^\transpose | \clY_t] 
\end{aligned}
\]
where again tildes represent deviations,  such as $\tilPhi^i_t = \Phi^i_t -\haPhi^i_t$.

 \Prop{t:EstIndExchangeable} states that some statistics of the individual  can be expressed in terms of those of the population:
It is \textit{not} the case that  $\Sigma_t^i=N \Sigma_t $,
since   $\{\Phi^i_t : 1\le i\le N\}$ are correlated.

\begin{proposition}
\label{t:EstIndExchangeable}
For each $s$, $t$, $i$,  and any set $S\subset\Re^d$, the   conditional probability is independent of $i$:
 \begin{equation}
\Prob\{ \Phi^i_s\in S \mid \clY_t\} = \Prob\{ \Phi^1_s\in S \mid \clY_t\} 
\label{e:EstIndExchangeable}
\end{equation}
and consequently,  $
\Expect[ \Phi^i_s \mid \clY_t]  = \Expect[ \Phi_s \mid \clY_t] $. 
Moreover,   
$\haPhi_{t+1\mid t}^i  = A_t \haPhi_t$,  the state covariances for the individual are
\[
\begin{aligned}
 \Sigma_{t}^i & = \diag(\haPhi_t) - \haPhi_t \haPhi_t^\transpose
\\
 \Sigma_{t+1\mid t}^i & = \diag(\haPhi_{t+1\mid t}) - \haPhi_{t+1\mid t} \haPhi_{t+1\mid t}^\transpose,
\end{aligned}
 \]
  and the cross covariances can be expressed, 
\[
\begin{aligned}
 \Expect\bigl[     \tilPhi^i_t   ( \tilPhi_t ) ^\transpose \mid \clY_t\bigr]  &= \Sigma_t
 \\
 \Expect\bigl[     \tilPhi^i_{t+1\mid t}   ( \tilPhi_{t+1\mid t} ) ^\transpose \mid \clY_t\bigr]  &= \Sigma_{t+1\mid t}   \,, \quad 1\le i\le N.
\end{aligned}
 \]
\end{proposition}

\begin{proof}
The proof of \eqref{e:EstIndExchangeable}
 follows from the symmetry and independence conditions imposed in (A1--A4).   The remaining results follow from this,  and the fact that $\Phi^i_t$ has binary entries [in particular,  $\Phi^i_t(\Phi^i_t)^\transpose = \diag(\Phi^i_t)$].
\end{proof}

Recall from the introduction that two formulations of the Kalman filter 
have been considered in this research. 
For a conditionally Gaussian model,  the Kalman filter equations require the conditional covariances for the state noise,
\begin{equation}
\!\!\!\!
\Sigma^{W^i}_t=\Expect[W^i_{t+1}(W_{t+1}^i)^\transpose \mid \clY_t] \,,
\ \
\Sigma^W_t=\Expect[W_{t+1}W_{t+1}^\transpose \mid \clY_t]
\label{e:NoiseCondCov}
\end{equation}
and also the conditional covariance of the measurement noise,
\begin{equation}
\Sigma^V_t=\Expect[V_t^2 \mid \clY_{t-1}] 
\label{e:NoiseObsCondCov}
\end{equation}
Formulae for the state noise covariances can be obtained in full generality.  We require the distribution of the random vector $\gamma_t$ introduced in A4 to obtain a formula for $\Sigma^V_t$.  

The Kalman filter that generates $L_2$-optimal estimates over all \textit{linear functions of the observations} uses instead the (unconditional) covariance matrices,
\begin{equation}
\barSigma^{W^i}_t=\Expect[W^i_{t+1}(W_{t+1}^i)^\transpose],
\quad
\barSigma^W_t=\Expect[W_{t+1}W_{t+1}^\transpose ],
\label{e:NoisCov}
\end{equation}
and $
\barSigma^{V}_t=\Expect[V_t^2]$  (the  notation used in standard texts is $Q_t =\barSigma^W_t $ and $R_t=\barSigma^{V}_t$    \cite{cai88}).   
We show in \Prop{t:NoiseCondCov} that the two covariance matrices in \eqref{e:NoiseCondCov} are   linear functions of the \textit{true conditional mean} $\haPhi_t$.   Expressions for the two covariance matrices in \eqref{e:NoisCov} follow from \Prop{t:NoiseCondCov} and the smoothing property of conditional expectation, provided we can compute $\Expect[\haPhi_t]=\Expect[\Phi_t]$.
The formula we obtain for $\Sigma^{V}_t$ in \Prop{t:NoiseObsCondCov}  is a linear function of the conditional covariance matrices $ \Sigma^i_{t+1\mid t}   $ and $  \Sigma_{t+1\mid t}    $.   It is unlikely we can obtain formula for the means of these covariance matrices, and hence we do not expect to obtain an exact formula for $\barSigma^{V}_t$.


\subsection{State noise covariance}
\label{s:KFstate}

 The following result provides formulae for the conditional covariances for the state noise \eqref{e:NoiseCondCov} as a function of the conditional mean $\haPhi_t$. 
\begin{proposition}
\label{t:NoiseCondCov}
Under Assumptions~A1--A4,
\begin{eqnarray} 
\Sigma^W_t &=&
 \frac{1}{N}\Bigl( \diag(A_t \haPhi_t) - A_t \diag(\haPhi_t) A_t^\transpose \Bigr)
 \label{e:NoiseCondCovForm}
 \\
\Sigma^{W^i}_t &=&
 \diag(A_t \haPhi_t^i) - A_t \diag(\haPhi_t^i) A_t^\transpose    
\label{e:NoiseCondCovForm-i}
\end{eqnarray}
The second covariance is independent of $i$,  with common value $\Sigma^{W^i}_t = N\Sigma^W_t$.
\end{proposition}

\medbreak

\begin{proof}
Since
$\{ W_t^i  : 1\le i\le N\}$ is uncorrelated, 
\begin{equation}
\Sigma^W_t=
 \frac{1}{N^2}\sum_{i=1}^N \Sigma^{W^i}_t 
 \label{e:SigmaWavg}
\end{equation}
Moreover, \Prop{t:EstIndExchangeable} gives  $\haPhi_t^i=\haPhi_t$, and from this or  \eqref{e:Phi-sum} it is obvious that
\[
\haPhi_t = \frac{1}{N} \sum_{i=1}^N \haPhi_t^i 
\]
Consequently,  \eqref{e:NoiseCondCovForm} follows from  \eqref{e:NoiseCondCovForm-i}.

The derivation of the formula  \eqref{e:NoiseCondCovForm-i} for $\Sigma^{W^i}_t $ is similar to the Kalman filter construction in \cite{lipkrirub84}.
Given the larger sigma-field,  
\[
\clY_t^+ = \sigma\{\Phi_r^i, Y_r, \zeta_r,  A_r : r\le t,\ i\le N\}
\]
the smoothing property of conditional expectation implies,
\[
\Sigma^{W^i}_t = \Expect\bigl[ \Expect[W_{t+1}^i (W_{t+1}^i)^\transpose \mid \clY_t^+] \mid \clY_t\bigr]
\] 
The inner conditional expectation is transformed using the definition  \eqref{e:Wi}:
\[
\begin{aligned}
\Expect[W^i_{t+1}(W_{t+1}^i)^\transpose &\mid \clY_t^+] 
\\
&=  
\Expect[  \Phi_{t+1}^i (\Phi_{t+1}^i)^\transpose   \mid \clY_t^+]  
\\
&\qquad -  \Expect[  \Phi_{t+1}^i     \mid \clY_t^+]   \Expect[  \Phi_{t+1}^i     \mid \clY_t^+]^\transpose
\\
&=  \diag(A_t \Phi_t^i)    
\\
&\qquad -  A_t \diag( \Phi_t^i) A_t^\transpose
\end{aligned}
\]
where the final equation uses
$ \Expect[  \Phi_{t+1}^i     \mid \clY_t^+] =A_t \Phi_t^i $, and the fact that $\Phi_r^i$ has binary entries for each $i$ and $r$.

Taking the conditional expectation given $\clY_t$ gives \eqref{e:NoiseCondCovForm-i}. 
\end{proof}

\subsection{Sampling and observation covariance}
\label{s:samplingCov}

The observation model used in the numerical experiments that follow is based on random sampling of loads:  An integer $n<N$ is held fixed, and at each time instant $t$ a distinct set of  $n$ indices $\{k_1,\dots, k_n\}$ is chosen uniformly at random.   The observation is the average 
\begin{equation}
Y_t= \frac{1}{n} \sum_{i=1}^n   \util(X^{ k_i}_t).
\label{e:samplingCov}
\end{equation}
   Recalling the definition $ \util(X^{ k_i}_t) = C\Phi^{k_i}_t $,  it follows from \Prop{t:EstIndExchangeable}  that this is an unbiased estimate of   $C\Phi_t $.  
An expression for the conditional variance $\Sigma_t^V$ defined in \eqref{e:NoiseObsCondCov} is obtained next.

\begin{proposition}
\label{t:NoiseObsCondCov}  
The conditional covariance is given by,
\[
\Sigma_{t+1}^V = \frac{1}{n}\frac{N-n}{N-1}  C \Bigl(  \Sigma^i_{t+1\mid t}    -  \Sigma_{t+1\mid t}    
\Bigr) C^\transpose 
\]
\end{proposition}

\medbreak

\begin{proof}
By definition,
\[
\Sigma^V_{t+1} = \Expect\Bigl[  \Bigl( \frac{1}{n} \sum_{i=1}^n   C\Phi^{ k_i}_{t+1} -  C\Phi_{t+1} \Bigr)^2 \mid \clY_t\Bigr]
\]
where the permutation is random, uniform, and independent of $\clY_t$.  It follows from symmetry of the model that we can consider just the first $n$ samples,
replacing $k_i$ by $i$.    We also center the random variables to obtain
\begin{equation} 
\Sigma^V_{t+1} = \Expect\Bigl[  \Bigl( \frac{1}{n}   \sum_{i=1}^n  C\tilPhi^{i}_{t+1\mid t} -  C\tilPhi_{t+1\mid t} \Bigr)^2 \mid \clY_t\Bigr]
\label{e:SigmaVdef} 
\end{equation}
The quadratic is expanded as the sum of three terms, 
\begin{equation}
\Sigma^V_{t+1} = S_1 - 2S_2 + S_3\,
\label{e:SigmaV3}
\end{equation}
in which the second two terms are transformed using the following consequence of \Prop{t:EstIndExchangeable}:
 \begin{equation} 
 \Expect\bigl[     \{ C\tilPhi^i_{t+1\mid t} \}  \{ C\tilPhi_{t+1\mid t} \}  \mid \clY_t\bigr] 
  = C\Sigma_{t+1\mid t} C^\transpose  
\label{e:crossPhi}
\end{equation}

The first term  is the conditional variance,
\begin{equation}
S_1 \eqdef
\Expect\bigl[  \bigl(   C\tilPhi_{t+1\mid t} \bigr)^2 \mid \clY_t\bigr] = C \Sigma_{t+1\mid t} C^\transpose
\label{e:SigmaVa}
\end{equation}
The second term is transformed using \eqref{e:crossPhi}:
\begin{equation} 
\begin{aligned}
S_2 \eqdef
 \Expect\Bigl[  \Bigl( \frac{1}{n}   &\sum_{i=1}^n  C\tilPhi^{i}_{t+1\mid t} \Bigr)\Bigl( C\tilPhi_{t+1\mid t} \Bigr) \mid \clY_t\Bigr]
 \\
  &=   C \Sigma_{t+1\mid t} C^\transpose
\end{aligned}
\label{e:SigmaVb}
\end{equation}
The third term is another conditional variance,
\[
\begin{aligned}
S_3 \eqdef
 \frac{1}{n^2}  
\Expect\Bigl[  & \Bigl( \sum_{i=1}^n  C\tilPhi^{i}_{t+1\mid t}   \Bigr)^2 \mid \clY_t\Bigr]
\\
  &=  \frac{1}{n^2}   \sum_{i,j=1}^n \Expect\bigl[      \{ C\tilPhi^{i}_{t+1\mid t} \}  \{ C\tilPhi^{j}_{t+1\mid t} \}    \mid \clY_t\bigr]
\end{aligned}
\]
Applying symmetry once more gives, 
\[
\begin{aligned} 
  \sum_{i,j=1}^n  \Expect\bigl[  &    \{ C\tilPhi^{i}_{t+1\mid t} \}  \{ C\tilPhi^{j}_{t+1\mid t} \}    \mid \clY_t\bigr]
  \\
  &= n  \Expect\bigl[      \{ C\tilPhi^1_{t+1\mid t} \}^2      \mid \clY_t\bigr] 
  \\
  &\quad + (n^2-n)  \Expect\bigl[     \{ C\tilPhi^{1}_{t+1\mid t} \}  \{ C\tilPhi^2_{t+1\mid t} \}    \mid \clY_t\bigr]
\end{aligned}
\]
Using familiar arguments, and applying \eqref{e:crossPhi},
\[
\begin{aligned}
 (N-1) &\Expect\bigl[     \{ C\tilPhi^{1}_{t+1\mid t} \}  \{ C\tilPhi^2_{t+1\mid t} \}    \mid \clY_t\bigr]
 \\
 &=
 \sum_{k=2}^N 
 \Expect\bigl[     \{ C\tilPhi^{1}_{t+1\mid t} \}  \{ C\tilPhi^{k}_{t+1\mid t} \}    \mid \clY_t\bigr]
 \\
  &= N  \Expect\bigl[     \{ C\tilPhi^{1}_{t+1\mid t} \}  \{ C\tilPhi_{t+1\mid t} \}   \mid \clY_t\bigr] 
  \\
  &\qquad   -   \Expect\bigl[     \{ C\tilPhi^{1}_{t+1\mid t} \}^2 \mid \clY_t\bigr]   
  \\
  & =
  N C\Sigma_{t+1\mid t} C^\transpose  - C\Sigma^i_{t+1\mid t} C^\transpose 
\end{aligned}
\]
Putting these expressions together gives,
\[
S_3 =
 \frac{1}{n^2}  C  \Bigl( n \Sigma^i_{t+1\mid t}   + \frac{n^2-n }{N-1} \Bigl[   N \Sigma_{t+1\mid t}   - \Sigma^i_{t+1\mid t} 
 \Bigr]
 \Bigr)C^\transpose
\]
The desired expression is obtained by substitution of these terms in \eqref{e:SigmaV3}.
\end{proof}

\subsection{Estimation of the individual }
 \label{s:EstInd}

A filter is introduced here to estimate the mean and variance for an individual state $\bfPhi^i$, and the mean and variance of the QoS metric $\Health^i_t$  for a typical load (recall   \eqref{e:QoS15}).
Estimates of the individual state and individual QoS can be used to estimate the flexibility of loads, which may vary with time.   For example, if the grid operator believes that every water heater contains water that is too cold,  then it is unlikely that there is much remaining flexibility to reduce power consumption from these loads.

%


The Kalman filter must be modified to obtain estimates of the QoS for individual loads.    One approach is to restrict to $\beta=1$,  and use the Markovian model  $(X^i_t,\Health^i_t)$.   This comes with high complexity:  In the case of the pool model,  the first component can take on $d=96$ values,   and the range of values of $\Health^i_t$ may be over 100.  This means that the state space is at least $10^4$,  which is probably  too large to be practical for estimation.

Here we introduce a Kalman filter for the joint process $\Psi^i_t=(\Phi^i_t;\Phi_t;\Health^i_t)$,  which is of dimension $2d+1$.    The construction of a linear model for $\bfPsi^i$ is based on (\ref{e:StateForEst},
\ref{e:StateForEst-i}),  and the one-dimensional dynamics for  QoS,
\[
\Health^i_{t+1} = \beta \Health^i_t  +  C_\health \Phi^i_t   
\]
where $ C_\health \Phi^i_t   =\health(X^i_t)$.
The measurements remain of the form \eqref{e:obs},  which is why it is necessary to include  $\Phi_t$ in the definition of $\Psi^i_t$.    
The new $A$, $C$, state-noise covariance matrices are,    
{\small
\[
\begin{bmatrix}
A_t & 0 & 0
\\
0&A_t & 0
\\
C_\health & 0 & \beta
\end{bmatrix}
\, ,
\quad 
[0,\, C, \, 0]
\, ,
\quad 
\begin{bmatrix}
\Sigma_t^{W^i} &N^{-1} \Sigma_t^{W^i} & 0
\\
N^{-1} \Sigma_t^{W^i} &\Sigma_t^{W} & 0
\\
0 & 0 & 0
\end{bmatrix}
\]}%
We thus have all of the system parameters needed to construct the Kalman filter.


\section{Estimating QoS and Closing the Loop}
\label{s:num}

We have conducted experiments in various settings, using both formulations of the Kalman filter, differentiated by the use of conditional or unconditional covariance matrices.   For estimating performance we obtained excellent results with the Kalman filter whose gain is based on the conditional covariance matrices.  The numerical results focus on the residential pool model, and closes with preliminary extension to TCL models.

The Kalman filter that is optimal over all linear estimators requires unconditional covariances. 
The state covariance matrix can be expressed as the mean of \eqref{e:NoiseCondCovForm}:
\[
\barSigma^W_t =
 \frac{1}{N}  \Expect\Bigl[ \diag(A_t \haPhi_t) - A_t \diag(\haPhi_t) A_t^\transpose \Bigr]
\]
The mean  $\Expect [\haPhi_t] = \Expect [\Phi_t] $ does not have a closed form expression, so we make two approximations.  First, we consider the mean with $\bfzeta\equiv 0$,  and second we let $t\to\infty$.  The covariance used in the filter is the resulting limit,
\begin{equation}
\barSigma^W_\infty =
 \frac{1}{N}   \Bigl[ \diag(A_t \pi_0) - A_t \diag(\pi_0) A_t^\transpose \Bigr]
\label{e:ssVarW}
\end{equation}
where $\pi_0$ is invariant for $P_0$.   A similar approximation is used for $\barSigma^V_t $.    

 
The common features in all of the numerical experiments that follow are listed here:   
The reference signal $\bfmr$ was 
generated from the BPA balancing reserves deployed in the first month of 2015 \cite{BPA}. This  signal was low pass filtered, and then scaled to an appropriate magnitude to match the capacity of the aggregate of loads, exactly as in \cite{meybarbusyueehr15}.  

The observation model was based on the sampling assumptions introduced in   
\Section{s:samplingCov}, in which a fixed number of loads $n<N$ is sampled at each time.  
Recall from \eqref{e:samplingCov} that the average of these $n$ samples at time $t$ is denoted $Y_t$.   
Unless stated otherwise, $n$ is taken to be $0.1\%$ of the total population.  
Hence, for a population of $N=10,000$ pools, exactly 10 are sampled at each time step.   

It is assumed that each pool  pump consumes 1kW while in operation.   
The actual power consumption at time $t$ is $NC \Phi_t$.  For this reason, plots of  $C \Phi_t$  or $\haY_t = C \haPhi_t$ are scaled by   $N$ to represent total power consumption or its estimate. 

Finally, the discount factor $\beta$ in the QoS metric \eqref{e:QoS15} was chosen as the value for which $\beta^k\approx 1/2$ when $k$ corresponds to one week.  The value  $\beta = 0.9997$  is obtained under the assumption of five-minute sampling.

\subsection{Estimation of QoS}
\label{s:QoS}

The Kalman filter described in \Section{s:EstInd} was used to obtain estimates of the mean and variance of the QoS metric $\Health^i_t$,  with
 $\health(x) = C x - \bary$.   These experiments were performed with $\bfzeta$ equal to the exogenous ``mean-field limit'' used in the linearized model of  \cite{chebusmey14}: 
\[
\zeta_t = \frac{G_c}{1+G_cG_p} r_t
\]
where $G_c$ is a PI compensator,  and $G_p$ is a transfer function for a linearized model
(see \cite{chebusmey14} for details).

 \begin{figure}[h]
\Ebox{.95}{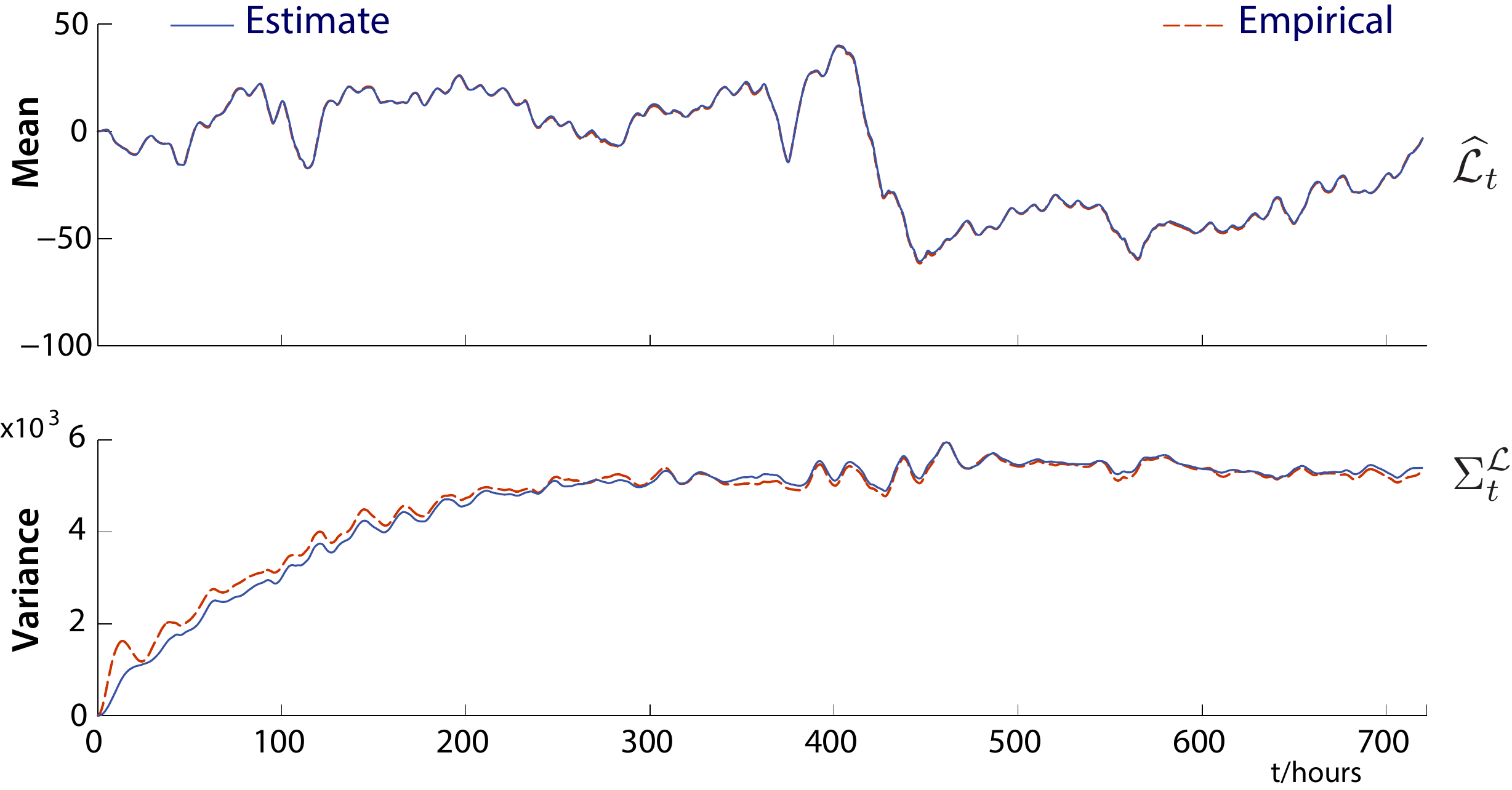}   
\vspace{-2.5ex}
\caption{Although the liner model is not observable, estimates of the mean and variance of QoS for an individual are nearly perfect.}
\label{f:obsQoS}
\vspace{-1ex}
\end{figure}

\Fig{f:obsQoS} shows estimates of QoS based on a model with 10,000 pools, with $0.1\%$ sampling rate.       The accuracy of estimation is remarkable, in spite of significant measurement noise.   

%

Non-ideal settings were also considered. In one experiment   a population of 20,000 pools was divided into two classes of equal size: 10,000 pools operated as previously with a
12hr/day cleaning cycle, 
and the other 
10,000
  operated with a 8hr/day cleaning cycle.   It is possible to construct the exact Kalman filter by doubling the dimension of the state space.   Instead,   an approximate model was constructed in which the parameters in the two state space models (differentiated by cleaning cycle) were averaged;  similar approximate models are used in  \cite{matkoccal13}.  
 
 \spm{shortened:}
Estimation results for the case of two pool classes are shown in \Fig{f:obsQoS8+12}. 
The QoS for the two classes of pools have different means and variances.   This difference cannot be captured using this filter which is designed to estimate the overall mean and variance.    The filter generates estimates of the mean QoS  that lie between the two empirical means,  and estimates of its variance follow more closely   the empirical variance of the pools with 12hr/day cleaning cycle. 

\begin{figure}[h]
\Ebox{.95}{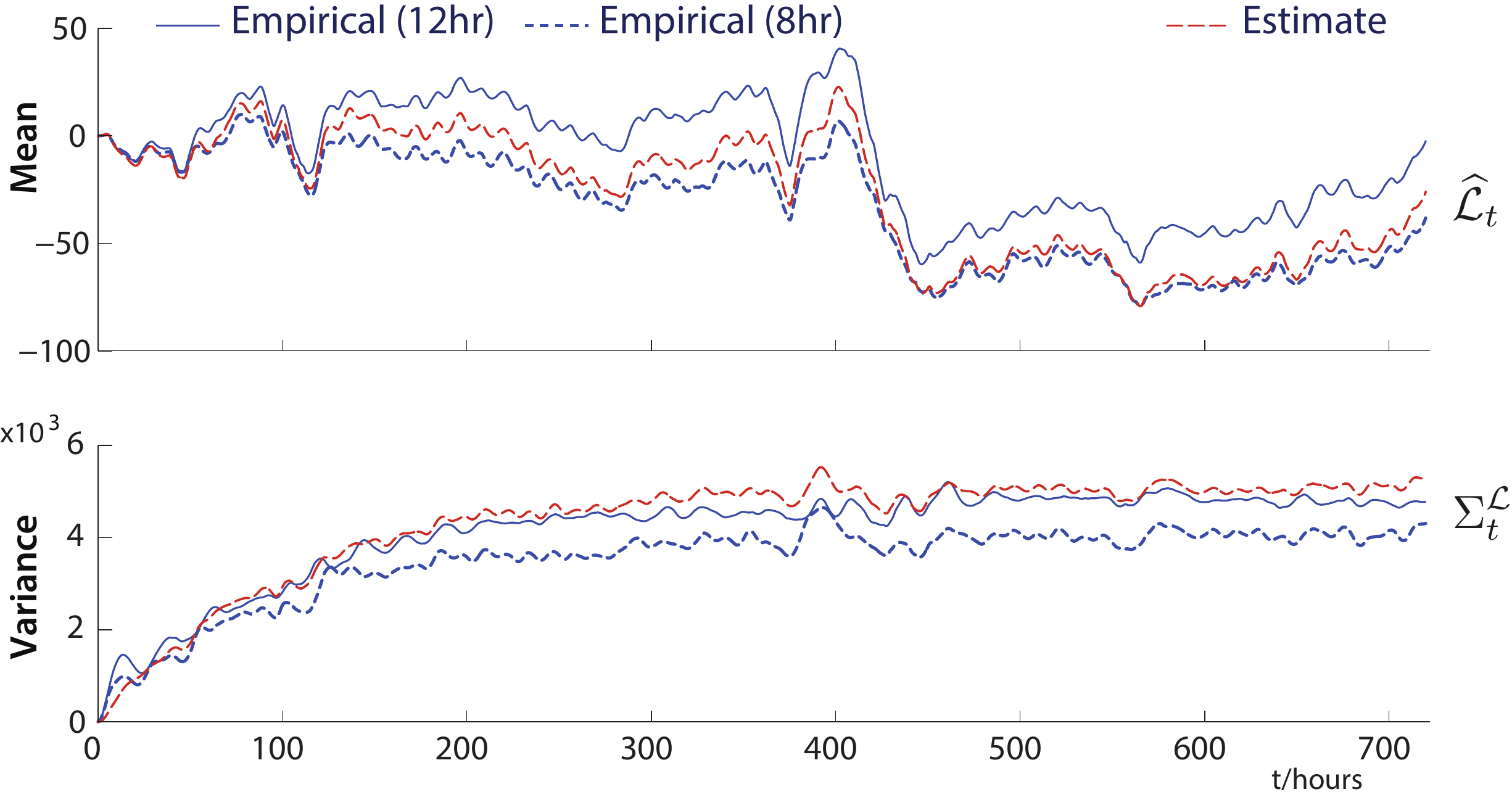}   
\vspace{-2.5ex}
\caption{Estimates of the mean and variance of QoS for an individual with a heterogenous population.   The un-modeled dynamics lead to some error.}
\label{f:obsQoS8+12}
\vspace{-1ex}
\end{figure}

\subsection{Control of the heterogeneous population}
\label{s:hetero}

Closed-loop experiments were conducted
in which un-modeled dynamics are present due to both
load heterogeneity and additional local control.  With large un-modeled dynamics it is still possible to obtain estimates of the mean of QoS,  but the Kalman filter cannot be expected to provide estimates of the variance.  

Error feedback was used in all of the remaining experiments,  $\zeta_t = G_c e_t$.
A PI compensator was used to define $G_c$:  A proportional gain of 50, and integral gain of 1.5 worked well in all examples. 
\spm{new:}
These values were based on a nominal LTI model --- this feedback design resulted in a phase margin of about 70 degrees, and gain margin of about 10dB.  

In the experimental results illustrated in
 \Fig{f:YfbPureObs}, the error signal was defined with respect to the raw measurements, without use of the Kalman filter.  In all other experiments,   the following definition
 was used:   
\begin{equation}
e_t= \frac{1}{N} r_t - \widetilde{Y}_t, \quad \widetilde{Y}_t = C \haPhi_t - \bary
\label{e:error}
\end{equation}
with $\bary$ equal to the average nominal power consumption.
The scaling of the reference by $N^{-1}$ is required since $C \haPhi_t $ is an estimate of average power consumption.  

The total number of pools was taken to be $N=300,000$.   Performance results using smaller values of $N$ are   discussed at the close of \Section{s:hetero}, and also in \Section{s:opt-out}.

The simulation model used in the remaining experiments consisted of nine different classes of pools, distinguished by cleaning cycle.  
  The distribution of classes is shown in Table~I.   Suppose that the grid operator had full information regarding the number of pools in each class.  
The full state space description would have dimension $9\times d$, which is far too complex.  Moreover, in practice the grid operator will not have complete system information.

Here we make some coarse approximations to simplify the model.  In addition to reducing complexity of the filter, our goal is to investigate robustness of the estimation algorithms in a closed-loop setting.

\section*{{\sc Table I: Distribution of pools}}
\vspace{.25em}
\begin{center} 
{\small
 \begin{tabular}{lccccccccc} 
    \hline\vspace{-.5em}
    \\ 
\hspace{-.5em}
        \textbf{Cycle} \ (hr/day)  & 20 & 18 & 16 & 14 & 12 & 10 & 8 & 6 & 4 \\[1ex] 
        \hline
        \\
\hspace{-.75em}
         \textbf{Number} \ ($\times 10^4$) & 1  & 1  & 1  & 3  & 5  & 10 & 5 & 3 & 1  \\[1ex]   \hline 
    \end{tabular}
    }
\end{center}

\smallbreak


Although the simulation uses these nine classes of pools,  for the purposes of filter design an  approximate model was constructed based on just two pool models, corresponding to 8 and 12-hour cleaning cycles.   
Following  \cite{matkoccal13}, a convex combination was chosen for the state matrix, 
\[
A_t = \alpha A_t^8 + (1-\alpha) A_t^{12}
\]
 in which  $0 \le \alpha \le 1$ is independent of $t$, and
 the superscripts $^{8}$ and $^{12}$ refer to the respective pool classes.

The parameter $\alpha$  was chosen to be consistent with the hypothesis that the collection of pools consist of just two classes, with 8 or 12 hr/day cycles.  It is assumed that the nominal average power consumption $\bary$ is known to the grid operator -- this can be estimated easily from weekly measurements.  The 2-class   assumption would imply 
that $\bary = \alpha \bary^8 +(1- \alpha) \bary^{12} $, giving 
$
\alpha =  (\bary^{12}-\bary)/(\bary^{12}-\bary^8)$,
with   $\bary^{12}=1/2$ and $\bary^8=1/3$.

 \spm{Yue: note revision to simplify, and avoid unsubstantiated conjectures. }
 
The  covariance matrix $\Sigma^{W}_t$ used in the Kalman filter was modified to take into account the large un-modeled dynamics.  It was chosen as the sum of two terms,
\begin{equation}
\Sigma^{W}_t = \Sigma^{W_*}_t + k_t \barSigma^W_\infty,  
\label{e:Cov_convex}
\end{equation}
where $k_t\ge 0$,
$\Sigma^{W_*}_t$ is a convex combination of the matrices computed in \Proposition{t:NoiseCondCov}, and
$\barSigma^W_\infty$ is a convex combination of the matrices  \eqref{e:ssVarW}.  This structure ensures that $\Sigma^{W}_t $ is positive semi-definite, and that $\Sigma^{W}_t \One = 0$ (consistent with the fact that our state evolves in the simplex).     The performance of the filter was not very sensitive to the parameter $k_t$,  but the best value was found to grow linearly with population size.  The results that follow used $k_t = N/300$   (note that $N$   cancels the factor $1/N$ appearing in the definition  \eqref{e:ssVarW} of    $\barSigma^W_\infty $).

\begin{figure}[h]
\Ebox{.95}{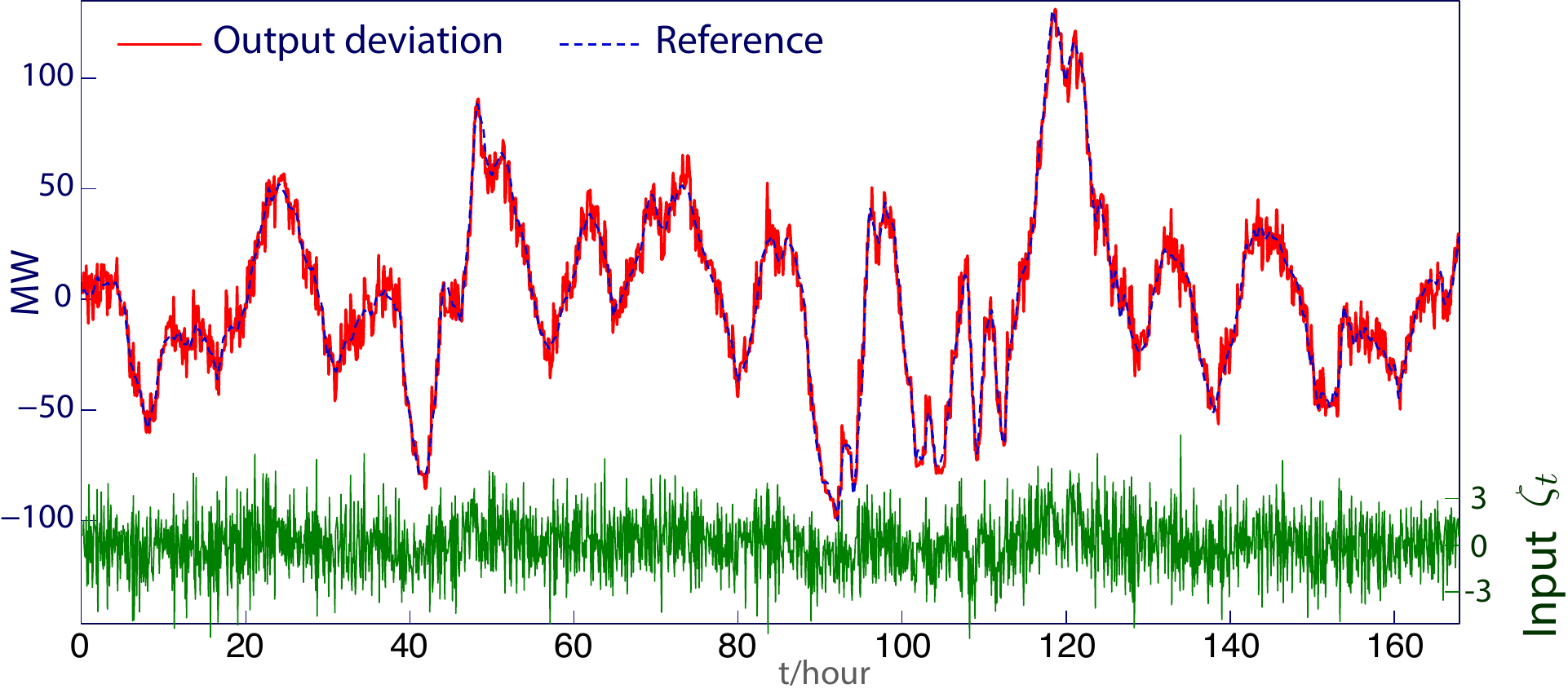} 
\vspace{-2.5ex}
\caption{PI control with output-error feedback, with $e_t= {N^{-1}} r_t - Y_t$.}
\label{f:YfbPureObs}
\vspace{-1.25ex}
\end{figure} 
With the experimental setting fully described, we now survey results from several control experiments.

We first show what happens  if  we forgo use of the Kalman filter,
and take $e_t = N^{-1} r_t - \tilY_t$,  with $\tilY_t = Y_t - \bary$,   and  $Y_t$ the noisy measurements obtained with 0.1\%\ sampling.  Results from one experiment are shown in \Fig{f:YfbPureObs}. The volatility of the measurements resulted in similar volatility of input $\zeta_t$. This drove the pools to switch frequently, and also resulted in degraded grid level tracking.   This motivates the use of a filter to smooth the measurements.

Results using the filtered error \eqref{e:error} are shown in \Fig{f:YfbKF}. This Kalman filter estimator based control gives nearly perfect tracking performance, in spite of the significant modeling error.
 
\begin{figure}[h]
\Ebox{.95}{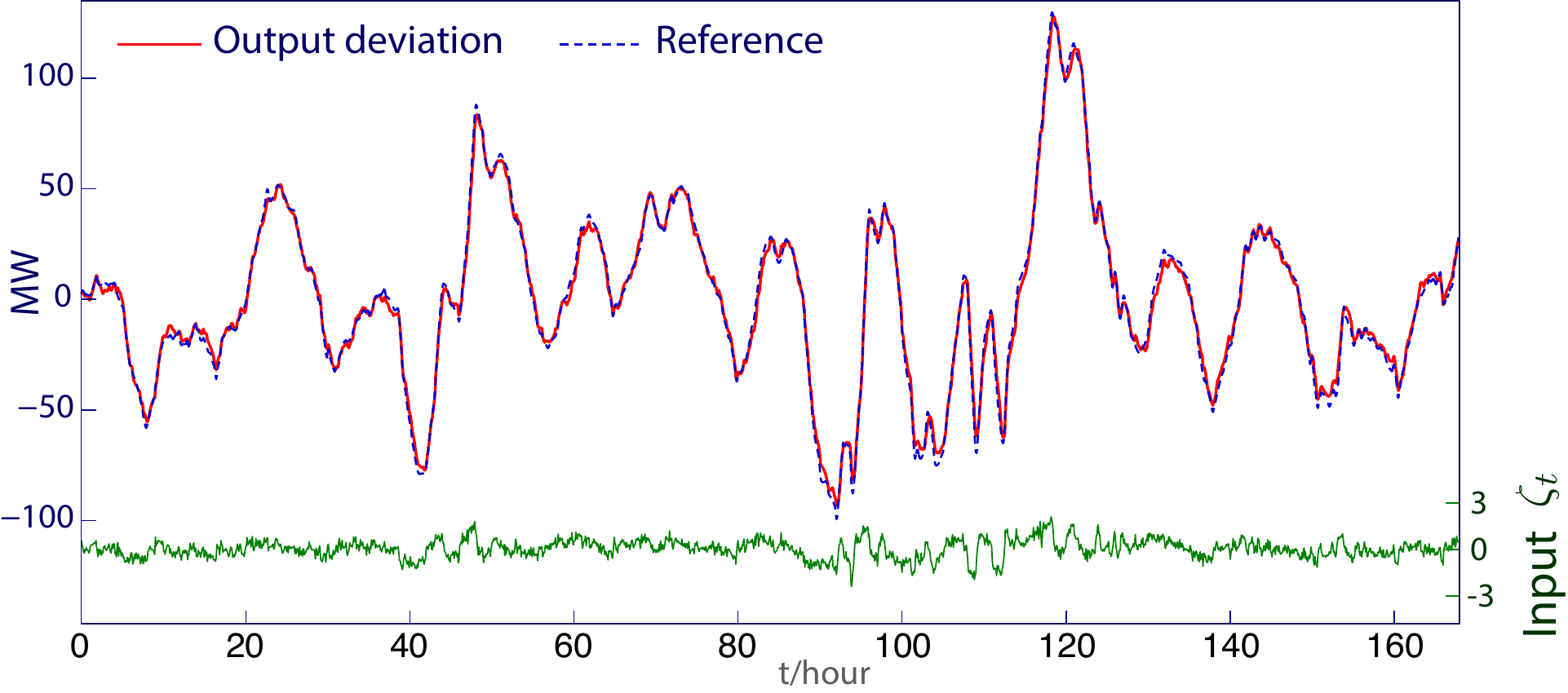} 
\vspace{-1.25ex}
\caption{PI control with smoothed error signal \eqref{e:error}.}
\label{f:YfbKF}
\end{figure}


%
%
%

\def\nrm{{\text{\rm n}}}

Tracking experiments were conducted in many other  scenarios, distinguished by sampling rate and load population, using the same distribution of load classes given in Table 1.   Normalized tracking error was chosen as the performance metric, 
\[
e_t^\nrm = \frac{ y_t - r_t }{ \| r \|_2},    \quad \| r \|_2 = \left( \frac{1}{T} \sum_{t=1}^T r_t^2 \right)^{\frac{1}{2}}.
\]

\begin{figure}[h]
\Ebox{.85}{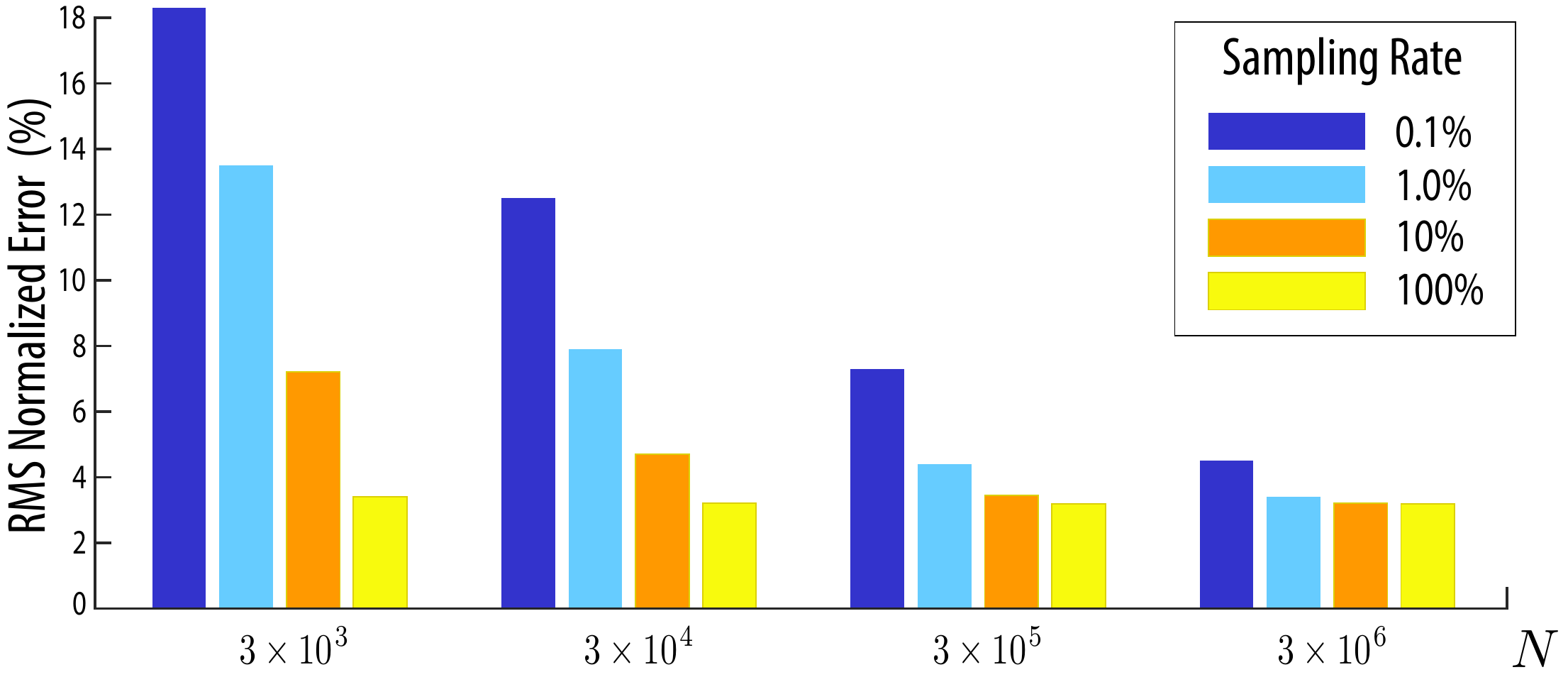}   
\vspace{-2.5ex}
\caption{Tracking performance improves with increased sampling rate, or a larger population of loads $N$.}
\label{f:bar}
\vspace{-1.25ex}
\end{figure}

 \begin{figure*}[t]
\Ebox{.85}{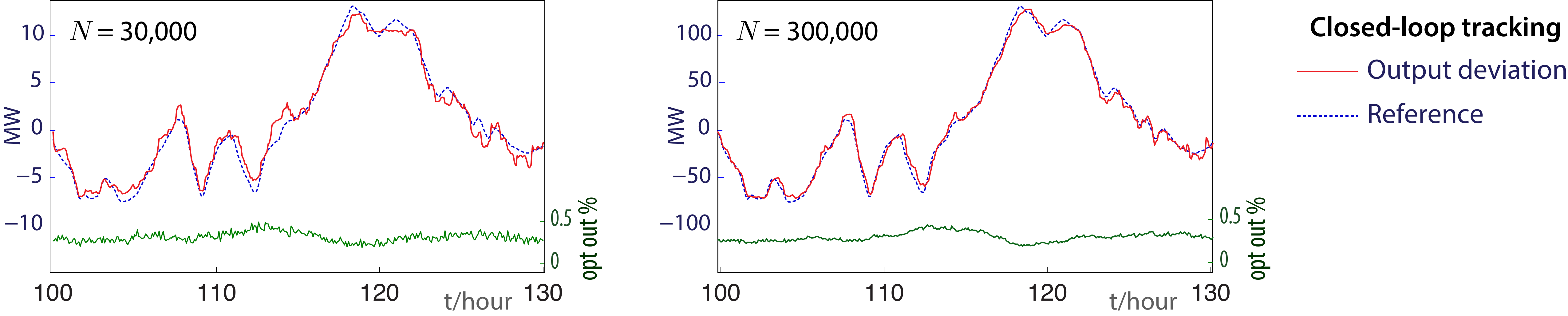}
\vspace{-2.5ex}
\caption{PI control and local opt-out control for the heterogeneous model with 0.1\%\ sampling rate.
Tracking performance improves as $N$ is increased.}
\label{f:track30K+300K}
\vspace{-1.25ex}
\end{figure*}

The full relationship between population size $N$, sampling rate, and closed loop performance is illustrated  in \Fig{f:bar}.  The vertical axis is the root mean square (RMS) of the tracking error $\{e_t^\nrm\}$.
 For each of the 4 values of $N$ considered,   the tracking error is less than $3.5\%$ with $100\%$ measurements.  It is found that the tracking error is best predicted by the number of samples per time-step.  Consequently, a smaller population of loads requires a larger sampling rate to maintain good tracking performance.

\subsection{Opt-out local control}
\label{s:opt-out}

In \cite{chebusmey14} an additional layer of control is introduced, based on the QoS metric $\Health^i_t $ introduced in \eqref{e:QoS15}:  The $i$th load can opt-out of service to the grid (ignore the signal $\zeta_t$) at  any time $t$ for which $\Health^i_t $ lies outside of pre-assigned bounds.  The impact of these additional  un-modeled dynamics is investigated here.

At each load the QoS metric $\Health^i_t$ is constrained to a common interval $[-50, +50]$: 
if $\Health_t^i \ge 50$, then the pool is turned off at time $t$, regardless of the value of $\zeta_t$. Similarly, if $\Health_t^i\le -50$  then it will be turned on. 


In the construction of the Kalman filter, the  state noise covariance was taken as \eqref{e:Cov_convex},  but in this case the scaling factor was chosen as a function of the QoS estimates,
$
k_t = k + b |\widehat{\Health}^i_t|
$.
This is chosen to reflect the fact that un-modeled dynamics increase with the percentage of pools that opt-out.   
The values  $k=N/100$ and $b=300$ worked well, and the sensitivity to these values was not high.


 
Closed loop  performance remained nearly perfect,  and the percentage of opt out loads was observed to remain under $0.5\%$ over the entire run.   \Fig{f:track30K+300K} shows results from two experiments with $N=300,000$  and $N=30,000$,  each with sampling rate maintained at 0.1\%. 

  It is seen again that  reducing the number of loads results in degradation of closed loop performance, since fewer loads are sampled at each time-step;  the two plots are nearly identical if the sampling rate is increased to 1\% for the smaller population.

%
%


\subsection{Observability in Demand Dispatch systems}
\label{s:symNot}

It is shown here that the lack of observability observed in the pool model with 12 hr cleaning cycle can be attributed to symmetry of the model.

A  review of linear systems terminology is required.
Consider the general LTI state-space model,  
\[
\Phi_{t+1} = A \Phi_t + B u_t,\qquad y_t=C \Phi_t
\]
where $A$ is a $d\times d$ matrix, $B$ is a $d$-dimensional column vector,
and $C$ is a $d$-dimensional row vector.   The observability matrix is the $d\times d$ matrix whose $k$th row is equal to $C A^{k-1}$,  and its null space  $\Sigma_{\bar o}$ is called the \textit{unobservable subspace}.   Suppose that $\{\Phi_t\}$ and $\{\Phi_t'\}$   are two solutions to the state space equations with common input sequence $\{u_t\}$, but different initial conditions,  $\Phi_0, \Phi'_0\in\Re^d$.  If  $\Phi'_0-\Phi_0\in\Sigma_{\bar o}$,   then for each $t\ge 0$,
\[
y_t - y_t' = C A^t (\Phi_0 -\Phi_0') = 0
\]
That is, based on measurements $\{y_t : t\ge 0\}$,  it is impossible to know if the true state sequence is $\{\Phi_t\}$ or $\{\Phi_t'\}$.   

For the time-varying system 
\[
\Phi_{t+1} = A_t\Phi_t + B u_t,\qquad y_t=C \Phi_t
\]
the observability matrix is replaced by the observability Grammian, $\clO_G$.  The definition can be found in any introductory text on state space control (or, see the encyclopedic text \cite{cai88}).   A critical interpretation is this:  If $\Phi_0 $,  $\Phi_0'$ are two initial conditions, then with identical inputs 
\[
(\Phi_0 -\Phi_0') ^\transpose  \clO_G(\Phi_0 -\Phi_0')   = \sum_{t=0}^\infty  (y_t - y_t')^2
\]
This is why the small eigenvalues observed in
\Fig{f:obsGramm} are a concern if we wish to reconstruct the entire state based on the observations $\{y_t\}$.

The proposition that follows concerns a class of LTI models that have the following symmetry property: \begin{equation}
A =  \begin{bmatrix}
A_g &  A_o
\\
A_o & A_g
\end{bmatrix}  \,, \qquad 
C  = [ \One^\transpose \mid 0^\transpose]  
\label{e:symLTI}
\end{equation}
where $A$ is a $d\times d$ matrix, and   the decomposition is in terms of two square matrices $A_g $, $ A_o$.   This requires that the integer $d$ is even.

An example is the pool pump model with 12 hour cleaning cycle.  Under certain conditions, a TCL model may be symmetric following a state transformation -- an example is given below.
Regardless of the details of the model, symmetry rules out observability.   

Analysis of the symmetric model is based on the  subspaces,
\[
\begin{aligned}
\clV &= \{ v\in\Re^d :  v_i= v_{d/2+i},\ 1\le i\le d/2\}
\\ 
\clV_0 &=\{ v\in \clV :  \One^\transpose v= 0\} 
\end{aligned}
\]
\begin{proposition}
\label{t:symNot}
Consider the symmetric LTI model \eqref{e:symLTI}.  If 
$A^\transpose $ is a transition matrix, then the subspace $ \clV_0$ is contained in the unobservable subspace. Consequently, the symmetric model is never observable, and  the rank of the observability matrix is no greater than $d/2+1$.
\end{proposition}

\begin{proof}
To prove the proposition we establish that $A\colon \clV_0\to\clV_0$, so that $A^k v\in\clV_0$ for any $k\ge 1$ and any $v\in\clV_0$.  By definition we also have $Cv=0$ for $v\in\clV_0$.
It follows that  $CA^kv=0$ for each $k$ when $v\in\clV_0$, so that $\clV_0\subseteq \Sigma_{\bar o}$ as claimed.

The subspace $\clV$ is equal to vectors of the form $v^\transpose=(z^\transpose \mid z^\transpose)$,
with $z\in \Re^{d/2}$,  and $\clV_0$ is the subset satisfying $\sum z_i=0$.    
Symmetry implies the identity
\[
Av = (w^\transpose \mid w^\transpose)^\transpose
\]
where $w = (A_g +  A_o) z$.   This shows that $A\colon \clV\to\clV$.

To complete the proof, we next show that $\sum w_i=0$ whenever $\sum z_i=0$, so that
$A\colon \clV_0\to\clV_0$.

Since $A^\transpose $ is a transition matrix, it follows that
\[
\One^\transpose =\One^\transpose A =  \begin{bmatrix}
\One^\transpose A_g+\One^\transpose  A_o \mid  \One^\transpose A_o + \One^\transpose A_g
\end{bmatrix}
\]
That is,   $Q= ( A_g +  A_o)^\transpose$ is a transition matrix:  $Q\One = \One$.
Consequently, if $z^\transpose \One = \sum_i z_i =0$, then
\[
\One^\transpose w = \One^\transpose (A_g +  A_o) z = z^\transpose Q \One = 0
\]
\end{proof}

Consider the following example in which the dimension of the symmetric state space model is $d=8$,  
\[
\begin{aligned}
 A_g &=\begin{bmatrix}
     1-\delta  & 0 & 0 & 0 \\
    a_{21} & 1-\delta  & 0 & 0 \\
    0 & a_{32} & 1-\delta  & 0 \\
 0 & 0 & a_{43} & 1-\delta  
\end{bmatrix}
\\
\text{\it and }\quad 
 A_0 &=\begin{bmatrix}
      a_{15} & a_{16} & a_{17} & \delta \\
      0 & 0 & 0 & 0\\
   0 & 0 & 0 & 0\\
   0 & 0 & 0 & 0
\end{bmatrix} 
\end{aligned}
\]
Based on the proposition,  the rank of the observability matrix is no greater than 5.

The following values are consistent with the  pool pump model with 12 hour cycle considered previously:   $\delta=1/6$,
$a_{21}=0.1654$, $a_{32}=0.0840$, $ a_{43} =0.0026$, $a_{15}=0.0013$, $a_{16}=0.0827$, $a_{17}=0.1641$. 
  The $8\times 8$ observability matrix $\clO$ was obtained using Matlab, and the {\tt null} command results in equality $\Sigma_{\bar o}=\clV_0$.  However, 
Matlab also returns these values for the log of the modulus of the eigenvalues:  $\{\log |\lambda_i|\} =$
\[
\{ 1.4 , 0.095, -1.8, -5.8, -19.7, -35.5, -36.5, -37.8\}
\]
Based on these values, it is not clear if the rank of $\clO$ is equal to 4, 5, or 8.  

\begin{figure}[h]
\Ebox{0.8}{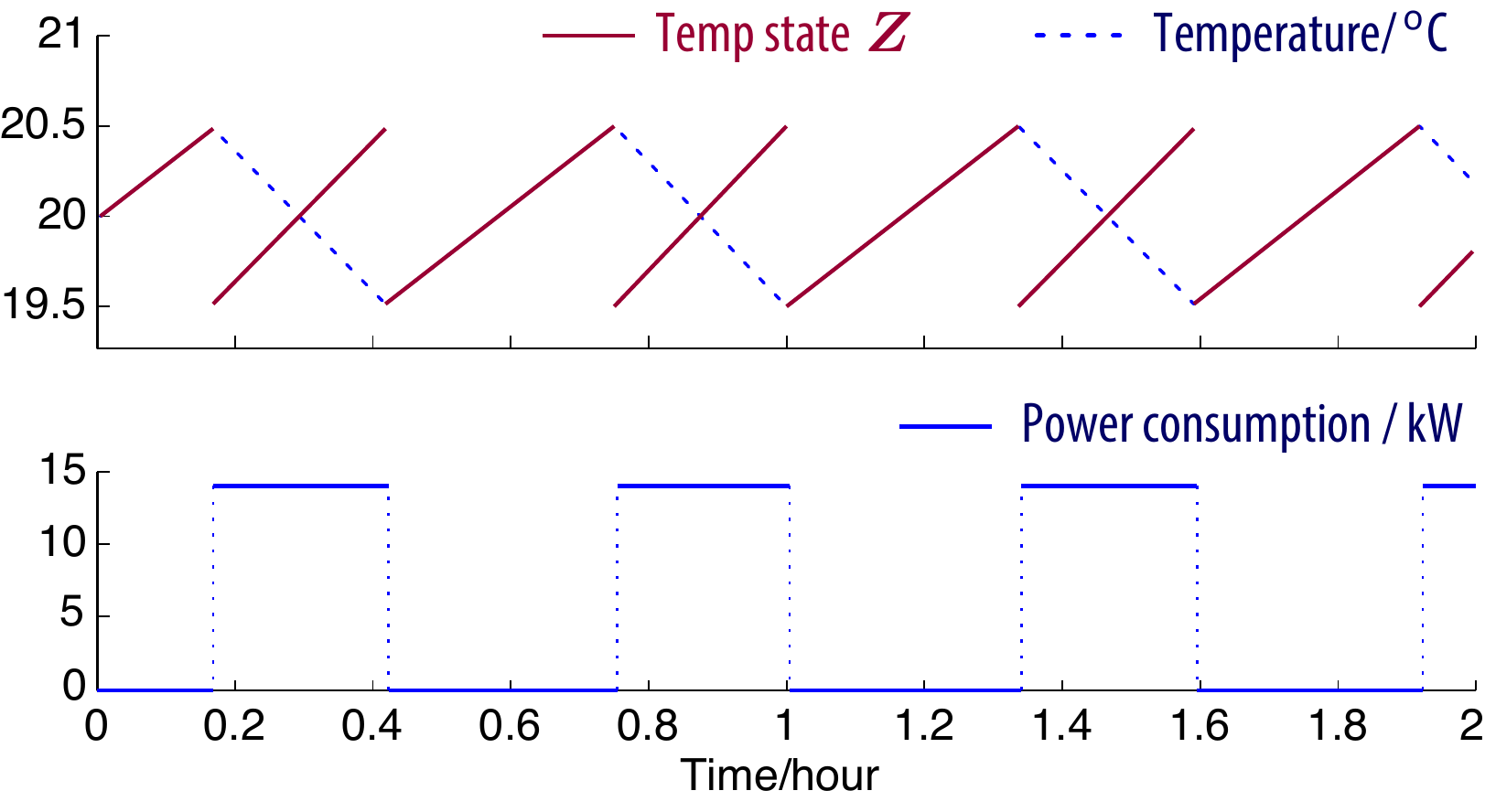} 
\vspace{-2.5ex}
\caption{State transformation for the TCL model: the evolution of $\bfmZ$ is similar to the pool model.}
\label{f:TCLtoPOOL}
\vspace{-1.25ex}
\end{figure}


Consider next a TCL that provides cooling  (a refrigerator or air-conditioner),  with temperature $\theta_t$ constrained to the dead-band $[\Thetamin , \Thetamax]$.   A new state process $(Z_t,m_t)$ is obtained in which $m_t$ is the power state as before, and
\[
Z_t =\begin{cases}   
			\theta_t  & \text{if $m_t=0$}
\\			
				\Thetamin +\Thetamax-\theta_t \quad & \text{if $m_t=1$}
	\end{cases}   
\]   
A typical trajectory for an air-conditioning unit shown in  \Fig{f:TCLtoPOOL} reveals that the sample paths are very similar to the pool  model.  If the Markovian dynamics of $\bfmZ$
are symmetric, then the LTI model is not observable.

\subsection{Eigenvectors and reduced order observer}
\label{s:redObs}

The structure of eigenvectors and reduced order observers are described for both the pool model and a  TCL model.

\paragraph*{Observability of the pool model}

We return to the homogeneous model with 12hr/day cleaning cycle.   The   observability Grammian for the time-varying linear system was found to be highly ill-conditioned in all of our experiments -- one example is illustrated  in  \Fig{f:obsGramm}.   The LTI model obtained with $\bfzeta\equiv 0$ is not observable, by an application of  \Prop{t:symNot}.

\begin{figure}[h]
\Ebox{.95}{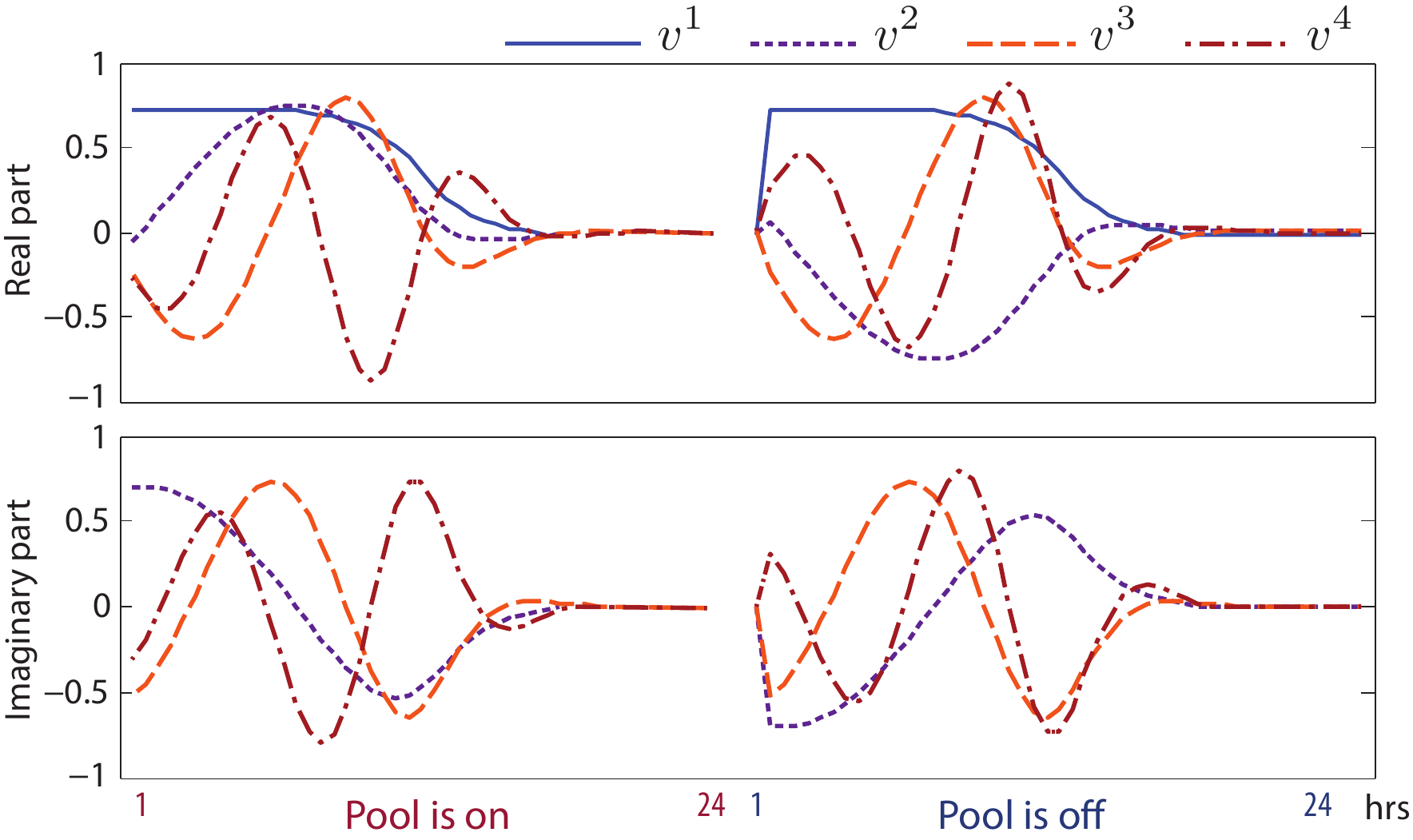} 
\vspace{-2.5ex}
\caption{First four  eigenvectors of $A$ for the pool model.}
\label{f:eig1-4}
\vspace{-1.25ex}
\end{figure}

\Fig{f:eig1-4} shows four eigenvectors of $A$  (left eigenvectors of $P_0$),  ordered with respect to the magnitude of the corresponding eigenvalues.
To avoid redundancy we considered only eigenvalues with non-negative imaginary part. 
The eigenvector $v^1$ has real eigenvalue $\lambda_1=1$:  it is 
proportional to
the invariant pmf  $\pi_0$, expressed as a column vector.   The remaining three eigenvalues are complex,  as are the corresponding eigenvectors $\{v^2,v^3, v^4\}$.  

A reduced-order observer was constructed using a state transformation based on the LTI model.   The construction can be interpreted as a truncated representation of the state with respect to the eigenvectors of $A$.  In view of the sinusoidal form shown in  \Fig{f:eig1-4},  this approximation is similar to a Fourier approximation.

  \begin{figure}[h]
\Ebox{.9}{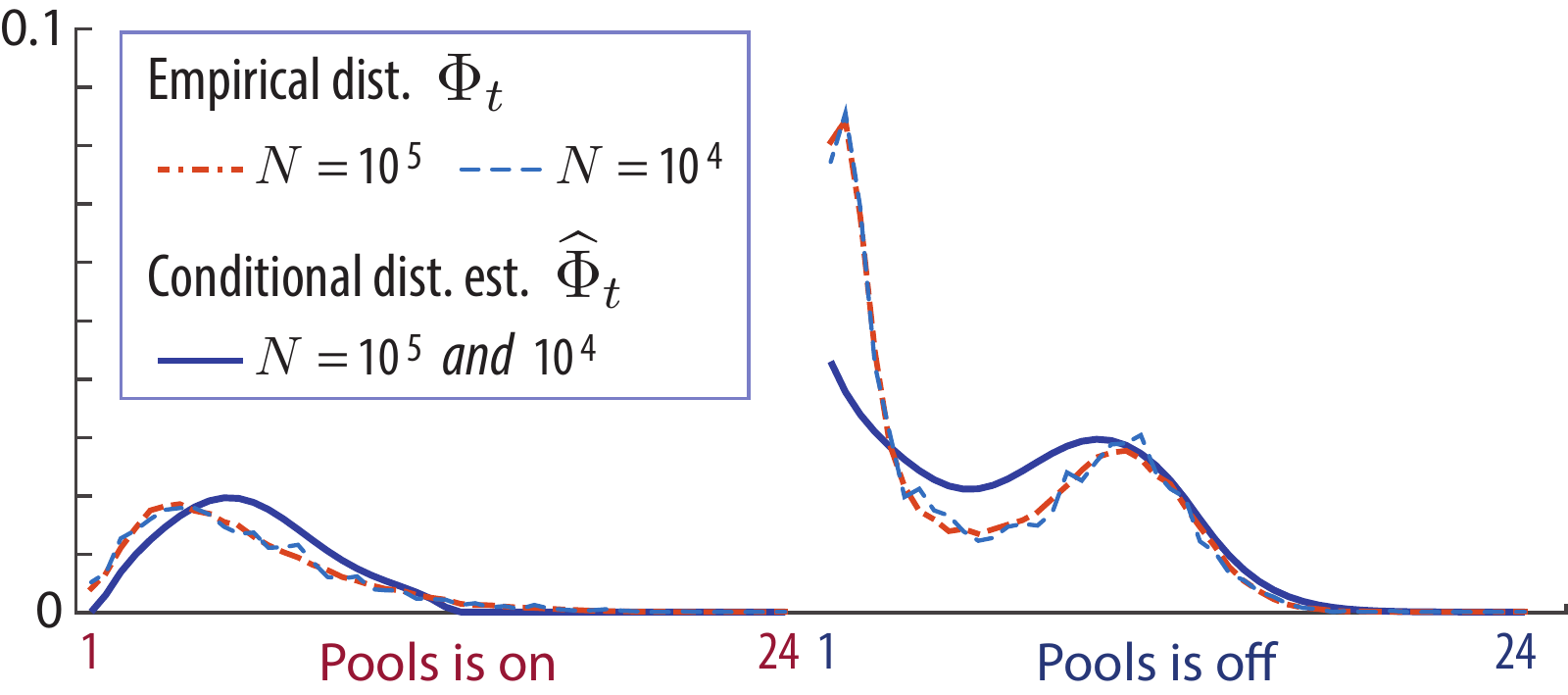}
\vspace{-2.5ex}
\caption{The  reduced-order observer yields good estimates of the state $\Phi_t$ for typical state values.  }
\label{f:RealAndCondDistOrder7}
\vspace{-1.25ex}
\end{figure}

After several experiments it was found that a 7th order approximation fit the dynamics well:
the  Bode plots nearly agreed, and the reduced order observer  tracked the true output $\util(X_t)$ nearly as well as the full order observer.  

Experiments were conducted for the model described in \Section{s:QoS},   with a 12hr/day cleaning cycle, and a sampling rate of  $0.1\%$.  \Fig{f:RealAndCondDistOrder7} shows a comparison of the state and its estimate at a particular time in two simulations, distinguished by load population.    

The empirical distribution $\Phi_t$ is approximately the same in each case,  $N=10^4$ and $N=10^5$;  it is slightly smoother for the larger population.   As seen in the figure, the state estimates are indistinguishable in the two cases.  The state $\Phi_t$  is far from the steady-state (proportional to $v^1$ shown in \Fig{f:eig1-4}),  yet the estimates are good in each experiment.  The explanation for success is most likely the similar continuity of the empirical distribution and the first seven eigenvectors.

%

\paragraph*{Observability of a TCL model}

The pool example is similar to the TCL model considered in  \cite{matkoccal13} and several other recent papers.   
The controlled Markov model for the $i$th load again consists of two parts:  a binary component indicating if the unit is on or off,  and a continuous component indicating its temperature.   
The state of the $i$th load at time $t$ is denoted $X^i_t = (m^i_t,\theta^i_t)$.   

The complete system description follows \cite{matkoccal13}:  For the nominal dynamics considered in prior work, this 
is defined by a dead-band $[\Thetamin , \Thetamax]$. 
For cooling devices,
\[ 
m^i_{t+1} = \left\{ 
  \begin{array}{l l l}
    0, & \quad  \theta^i_{t+1} < \Thetamin\\
    1, &  \quad \theta^i_{t+1} > \Thetamax\\
    m^i_t, &  \quad \text{otherwise}
  \end{array} \right.
\]  
The temperature is modeled as a linear system driven by white noise:
 \[
 \theta^i_{t+1} = a^i \theta^i_t + (1-a^i)(\thetamb_t - m^i_t R^i \ptcl) + \eta^i_t,
 \]
in which $0<a^i<1$.


The parameters  summarized in Table~II are taken from  \cite{matkoccal13}.  In this prior work a heterogeneous model was considered, in which the values $(a^i,R^i,C^i)$ were   sampled from a probability distribution.  A homogeneous model is considered in the experiments surveyed here, in which $R^i=C^i=2$ for each $i$.

\section*{{\sc Table II: TCL Parameters}}
\vspace{-.5em}
\begin{center}  
{\small
 \begin{tabular}{ll}
    \hline
\vspace{-.5em}
    \\
\hspace{-.5em}
        \textbf{Parameter}  & \textbf{Value} \\[1ex] 
        \hline
        \\[.1ex]
 $\tau$:  time step (secs) & 2   
 \\
$[\Thetamin , \Thetamax]$: temperature deadband  \ ($^\text{o}$C) & [20, 21]
        \\
      $\thetamb$: ambient temperature  \ ($^\text{o}$C) & 32
        \\
      $R$: thermal resistance  \ ($^\text{o}$C/kW) & 2
        \\
      $C$: thermal capacitance  \ (kWh/$^\text{o}$C) & 2
        \\
      $\eta$: modeling noise  \ ($^\text{o}$C)& $\sim N(0, 2.5\times 10^{-7}$)
        \\
      $\ptcl$: energy transfer rate  \ (kW) & 14
      \\[.5ex]
$a^i = \exp\{-\tau/(CR) \}$ &  
        \\[1ex]   \hline 
    \end{tabular}
    } 
\end{center}

\smallbreak

The state space for the TCL model is continuous.  A finite state-space approximation is obtained by binning the dead-band interval into 40 bins of equal size.  A Markov chain model is desired with a total of 80 states,  since the state captures temperature as well as whether the load is on or off.   

An empirical model was obtained via Monte-Carlo. Simulation experiments with 10,000 TCLs and 3,600 time steps  were run to obtain an approximation of the steady-state joint distribution,
\[
\Pi(x^j,x^k) = \Prob\{X_t=x^j,\ X_{t+1}=x^k\},\quad 1\le j,k \le 80.  
\]
The approximate transition matrix is then obtained using Bayes' rule:
\[
P_0(x^j,x^k)  = \frac{\Pi(x^j,x^k) }{\pi(x^i)},
\ \  \text{
with $\pi(x^i)=\sum_k  \Pi(x^j,x^k) $.  }
 \]
 \spm{Nov 20, 2015 (evening in Gainesville):  I just noticed a small notational issue. We should write $A_{kj} = P_0(x^j,x^k) $,
 just as we write $C_j = \util(x^j)$.}

The approximate Markovian dynamics are used in  \cite{matkoccal13} to define the heuristic mean-field model,  
\[
\Phi_{t+1} = A \Phi_{t},\qquad \{A_{kj}\} = \{P_0(x^j,x^k) \}.
\]
The observation matrix used in  \cite{matkoccal13} is the same as in the pool example, except for a scaling: 
\[
C_k =  10^5 \times \ptcl \times \ind\{\text{State $k$ is ON}\}. 
\] 
The observability Grammian for this LTI system is highly  ill-conditioned:
\Fig{f:RealAndCondDistOrder7_TCL}
shows that the rapid decay of the eigenvalues is similar to what was observed in \Fig{f:obsGramm} for the pool model.

\spm{Let's leave out 'rank' here since this depends on tolerance (I explain this earlier)
\\
The rank command in Matlab applied to the observability Grammian for this LTI system resulted in a value of 32, meaning that the system is not observable.   
}

\spm{Yue:  Do we need to make any more changes to the TCL section?  Also,  I believe 'Grammian' is far more common than 'Grammian'!  Look it up on Google.  Only Wikipedia insists on one 'm' \\
response: I think this section is good: we described how we obtained the A and C matrices of TCL model, and provided eigenvalues of its Observability Grammian.
}
\begin{figure}[h]
\Ebox{0.7}{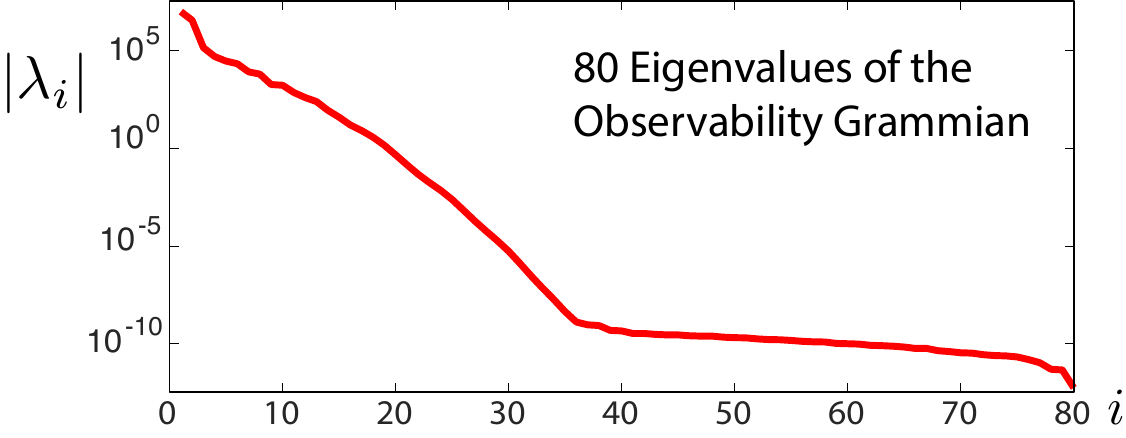} 
\vspace{-2.5ex}
\caption{Magnitude of eigenvalues of observability Grammian for TCL model}
\label{f:RealAndCondDistOrder7_TCL}
\vspace{-1.25ex}
\end{figure} 

The structure of eigenvectors is also similar. 
\Fig{f:eig1-4_TCL} shows four eigenvectors of the  matrix $A$,
corresponding to the four  eigenvalues of maximum magnitude (ignoring duplication from
complex conjugate pairs).    

\begin{figure}[h]
\Ebox{.95}{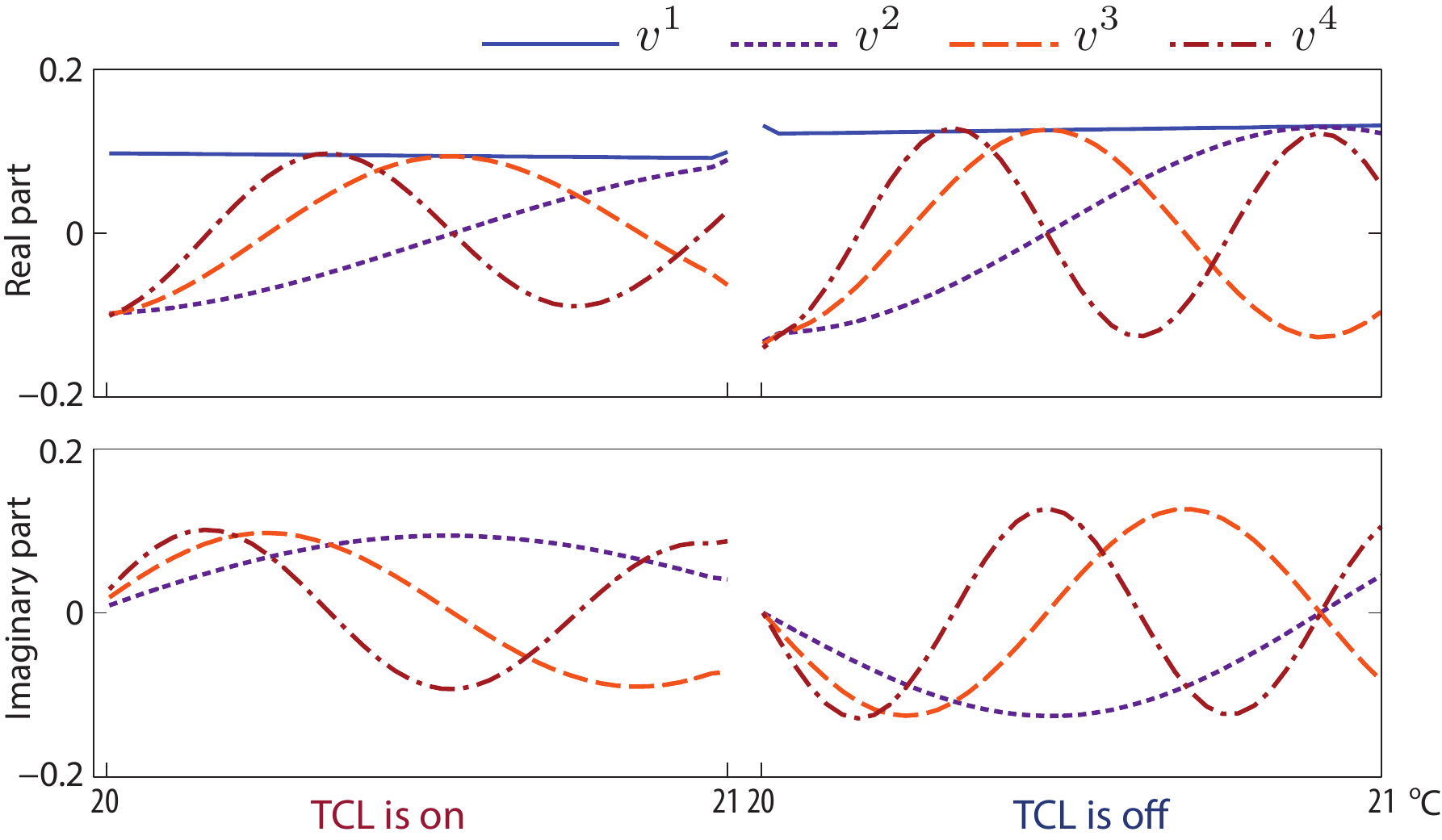} 
\vspace{-2.5ex}
\caption{First four  eigenvectors of $A$ for a TCL model.}
\label{f:eig1-4_TCL}
\vspace{-1.25ex}
\end{figure} 

A finer analysis of the continuous time / continuous state model for TCLs (as considered in  
\cite{lumThesis15,malcho88}) may give greater insight on the eigenstructure observed for these models.

 \section{Conclusions}
 \label{s:con}
 
 The construction of a Kalman filter for the joint population/individual dynamics is possible, and performance is remarkable in the test cases considered.   In particular, it is very surprising to obtain accurate tracking of both the mean and variance of QoS for an individual, given extremely noisy estimates of the population.

An open topic for future research is the state estimation techniques that take into account  opt-out, 
which is used to ensure good QoS to loads  \cite{chebusmey14}.  It may be possible to obtain a reduced order observer
for the very complex model for joint state-QoS dynamics.   Alternatively, it may be possible to exploit the structure of the distribution of QoS observed in numerical experiments.   When subject to opt-out control, the histogram of QoS is similar to a conditional Gaussian that is constrained to the interval specified by opt-out parameters.  A nonlinear filter may be constructed that takes this structure into account.

\smallbreak

\bibliographystyle{IEEEtran}
\bibliography{strings,markov,q}


\vspace*{-1.5\baselineskip}

\begin{IEEEbiography}[%
{\includegraphics[width=1in,height=1.25in,clip,keepaspectratio]{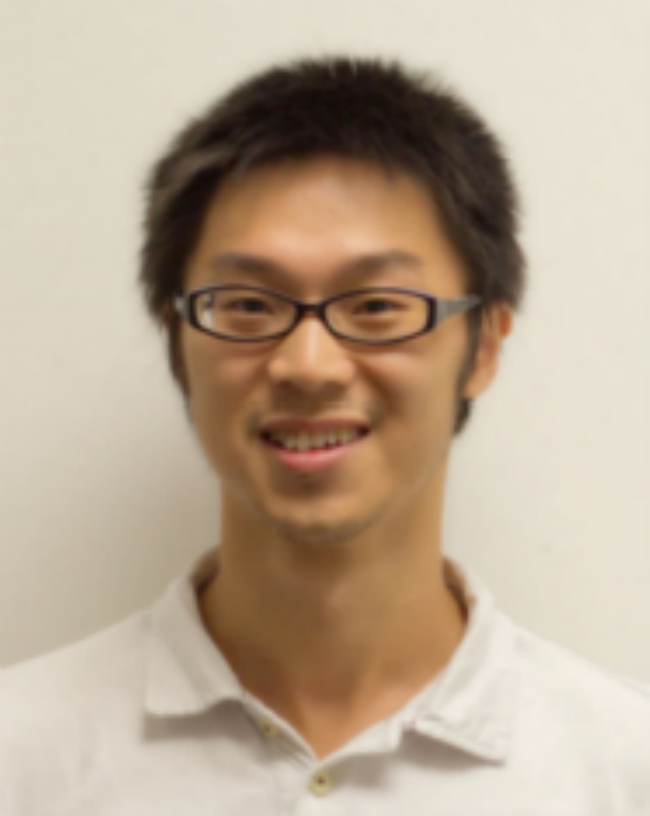}}%
]{Yue Chen}
 (S' 14) received the B.E. degree in electronic and information engineering from Harbin Engineering University, Harbin, China, in 2008, the M.E. degree in electrical engineering from University of Detroit Mercy, Detroit, USA, in 2010, the M.S. degree in Electrical and Computer Engineering in 2012 from University of Florida, Gainesville, USA.  During the Fall of 2015 he was an intern at Los Alamos National Laboratory
conducting research on control techniques for  demand dispatch. 
 In 2016 he is expected receive the Ph.D.\ degree in electrical and computer engineering at University of Florida.
\end{IEEEbiography}
\begin{IEEEbiography}[%
{\includegraphics[width=1in,height=1.25in,clip,keepaspectratio]{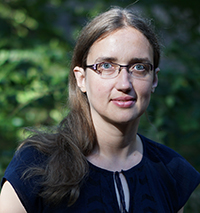}}%
]{Ana Bu\v{s}i\'{c}}
is a Research Scientist at Inria Paris -- Rocquencourt, and a member of the Computer Science Department at Ecole Normale Sup\'erieure, Paris, France. She received a M.S. degree in Mathematics in 2003, and a Ph.D. degree in Computer Science in 2007, both from the University of Versailles. She was a post-doctoral fellow at Inria Grenoble -- Rh\^one-Alpes and at University Paris Diderot-Paris 7. She joined Inria in 2009. She is a member of the Laboratory of Information, Networking and Communication Sciences, a joint lab between Alcatel-Lucent, Inria, Institut Mines-T\'el\'ecom, UPMC Sorbonne Universities, and SystemX. Her research interests include stochastic modeling, simulation, performance evaluation and optimization, with applications to communication networks and energy systems. 
\end{IEEEbiography}
\begin{IEEEbiography}[%
{\includegraphics[width=1in,height=1.25in,clip,keepaspectratio]{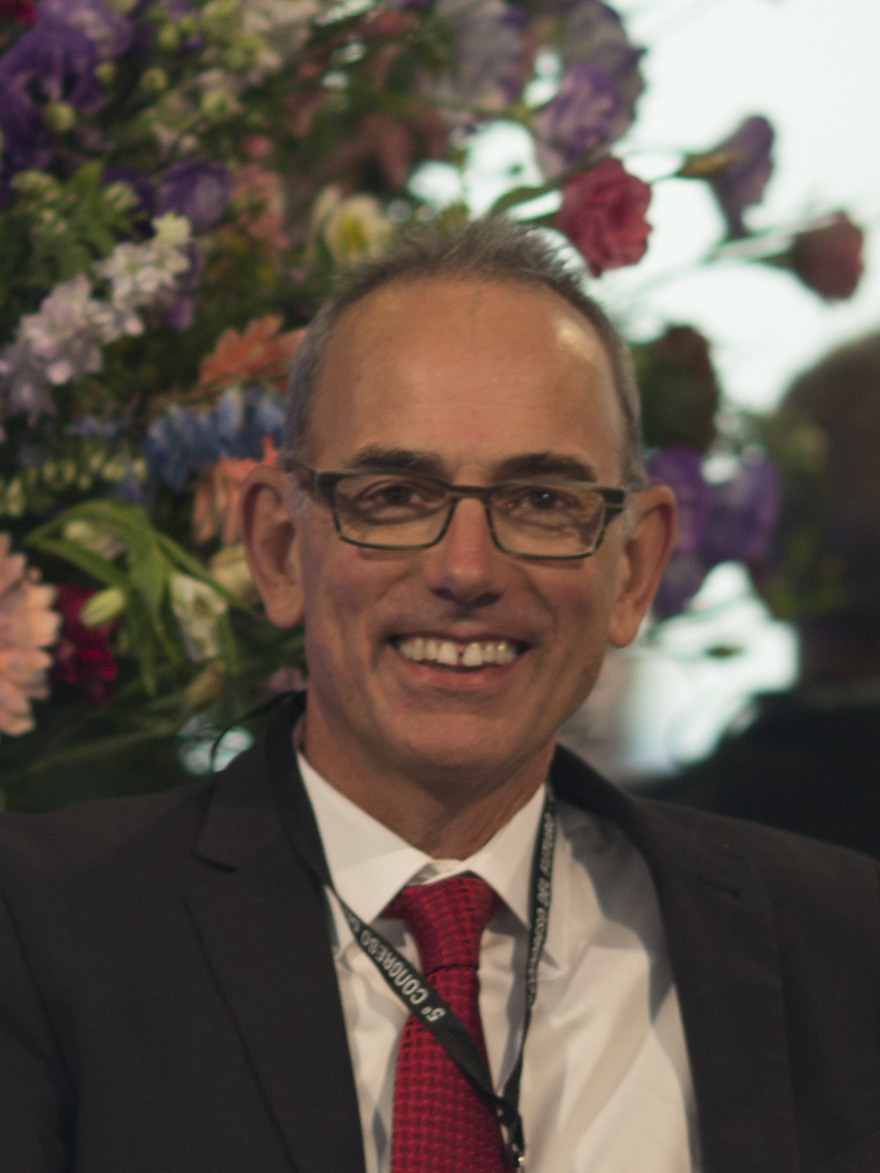}}%
]{Sean P.~Meyn}
(S'85, M'87, SM'95, F'02)
received the B.A. degree in Mathematics Summa Cum Laude from UCLA in 1982, and
the PhD degree in Electrical Engineering from McGill University in 1987 (with
Prof. P. Caines).  He held a  two year postdoctoral fellowship at the
Australian National University in Canberra, and was a professor at the
University of Illinois from 1989-2011.  Since January, 2012
he has been a professor and has held the Robert C. Pittman Eminent Scholar
Chair in Electrical and Computer Engineering at the University of Florida.
He is coauthor with Richard Tweedie of the monograph Markov Chains and
Stochastic Stability, Springer-Verlag, London, and received jointly with
Tweedie the 1994 ORSA/TIMS Best Publication In Applied Probability Award.
His research interests include stochastic processes, optimization, complex
networks, and information theory, with applications to smarter cities and
smarter grids.
\end{IEEEbiography}

\null  

\end{document}